\documentclass[journal]{IEEEtran}

\usepackage{cite}

\usepackage{amsmath, amsthm, amssymb}
\usepackage{url}

\usepackage{booktabs} 

\usepackage{tikz}
\usetikzlibrary{automata,positioning,matrix,shapes.callouts,patterns}

\usepackage[all,pdftex]{xy}
\CompileMatrices 
\xyoption{v2}
\xyoption{curve}
\xyoption{2cell}
\SelectTips{cm}{}  
\UseAllTwocells

\usepackage{xcolor}
\definecolor{dgreen}{rgb}{0, .6, 0}

\usepackage{pifont}
\newcommand{\cmark}{\ding{51}}%

\usepackage{centernot}
\usepackage{xspace}

\newcommand{\R}{{\mathbb{R}}}
\newcommand{\N}{{\mathbb{N}}}
\newcommand{\Rnn}{\R_{\ge 0}}
\newcommand{\Rpos}{\R_{> 0}}
\newcommand{\Z}{{\mathbb{Z}}}
\newcommand{\Znn}{\Z_{\ge 0}}
\newcommand{\Zp}{\Z_{> 0}}
\newcommand{\str}{w}

\newcommand{\currSN}{\mathit{currSN}}
\newcommand{\nextSN}{\mathit{nextSN}}
\newcommand{\currS}{\mathit{currS}}
\newcommand{\nextS}{\mathit{nextS}}

\newcommand{\pass}{\mathsf{pass}}
\newcommand{\mask}{\mathsf{mask}}

\newcommand{\Match}{\mathrm{Match}}
\newcommand{\myparagraphdense}[1]{\subsection{#1}}
\newcommand{\myparagraph}[1]{\subsection{#1}}
\newcommand{\stuck}[1]{s_{0}\centernot{\xrightarrow{#1}}\cdot}
\newcommand{\monaa}{${\mathtt{MONAA}}$\xspace}

\usepackage{gnuplot-lua-tikz}
\usepackage{pgfplotstable} 
\pgfplotsset{compat=1.12}

\newif\ifignore 
\ignorefalse
\newcommand{\auxproof}[1]{
\ifignore\mbox{}\newline
\textbf{BEGIN: AUX-PROOF} \dotfill\newline
{#1}\mbox{}\newline
\textbf{END: AUX-PROOF}\dotfill\newline
\fi}

\usepackage{wrapfig}
\theoremstyle{plain}
\newtheorem{theorem}{Theorem}[section]
\newtheorem{lemma}[theorem]{Lemma}
\newtheorem{proposition}[theorem]{Proposition}

\theoremstyle{definition}
\newtheorem{definition}[theorem]{Definition}

\theoremstyle{definition}

\newtheorem{example}[theorem]{Example}
\newtheorem{myexample}[theorem]{Example}

{\itshape}{\rmfamily}

\theoremstyle{plain}
\newtheorem{mynotation}[theorem]{Notation}

\makeatletter
\newcommand{\Biggg}{\bBigg@{4}}
\newcommand{\Bigggg}{\bBigg@{5}}
\newcommand{\Biggggg}{\bBigg@{6}}
\makeatother

\usepackage{wrapfig}
\usepackage{algorithm}
\usepackage[noend]{algpseudocode} 
\usepackage{etoolbox}

\makeatletter
\newcommand*{\algrule}[1][\algorithmicindent]{%
   \makebox[#1][l]{%
       \hspace*{.2em}
       \vrule height .75\baselineskip depth .25\baselineskip
   }
}
\algrenewcommand\algorithmicindent{1.3em}%

\newcount\ALG@printindent@tempcnta
\def\ALG@printindent{%
    \ifnum \theALG@nested>0
    \ifx\ALG@text\ALG@x@notext
    \else
    \unskip
    \ALG@printindent@tempcnta=1
    \loop
    \algrule[\csname ALG@ind@\the\ALG@printindent@tempcnta\endcsname]%
    \advance \ALG@printindent@tempcnta 1
    \ifnum \ALG@printindent@tempcnta<\numexpr\theALG@nested+1\relax
    \repeat
    \fi
    \fi
}
\patchcmd{\ALG@doentity}{\noindent\hskip\ALG@tlm}{\ALG@printindent}{}{\errmessage{failed to patch}}
\patchcmd{\ALG@doentity}{\item[]\nointerlineskip}{}{}{} 
\makeatother

\usepackage{graphicx}	
\usepackage{version}	

\usepackage{tikz}
\usepackage{textcomp}
\usepackage{hyperref}
\usepackage{lipsum}

\newcommand\copyrighttext{%
  \footnotesize \textcopyright 2018 IEEE. 
  DOI: \href{https://doi.org/10.1109/TCAD.2018.2857358}{10.1109/TCAD.2018.2857358}}
\newcommand\copyrightnotice{%
\begin{tikzpicture}[remember picture,overlay]
\node[anchor=south,yshift=10pt] at (current page.south) {\fbox{\parbox{\dimexpr\textwidth-\fboxsep-\fboxrule\relax}{\copyrighttext}}};
\end{tikzpicture}%
}

\begin{document}
%
\title{Moore-Machine Filtering  for\\ Timed and Untimed Pattern Matching}
%
%
%
\author{Masaki~Waga and Ichiro~Hasuo
\thanks{This article was presented in the International Conference on Embedded Software 2018 and appears as part of the ESWEEK-TCAD special issue.}
\thanks{This work is supported by JST ERATO
HASUO Metamathematics for Systems Design Project (No.\ JPMJER1603), 
and JSPS Grants-in-Aid No.\ 15KT0012 \& 18J22498.}
\thanks{M. Waga and I. Hasuo are with National Institute of Informatics,
Tokyo, Japan,
and SOKENDAI (the Graduate University for Advanced Studies), Kanagawa, Japan.}
\thanks{M. Waga is a JSPS Research Fellow.}
}

%
%

\markboth{IEEE transactions on Computer-Aided Design of Integrated Circuits and Systems}%
{Waga \MakeLowercase{\textit{et al.}}: Moore-Machine Filtering for Timed and Untimed Pattern Matching}
%



\maketitle
\copyrightnotice

\begin{abstract}
\emph{Monitoring} is an important body of techniques in runtime verification of real-time, embedded and cyber-physical systems. Mathematically,  the monitoring problem can be formalized as a  pattern matching problem against a pattern automaton. Motivated by the needs in embedded applications---especially the limited channel capacity between a sensor unit and a  processor that monitors---we pursue the idea of \emph{filtering} as preprocessing for monitoring. Technically, for a given pattern automaton, we present a construction of a \emph{Moore machine} that works as a filter. The construction is automata-theoretic, and we find the use of Moore machines particularly suited for embedded applications, not only because their sequential operation is relatively cheap but also because they are amenable to hardware acceleration by dedicated circuits. We prove soundness  (i.e.\ absence of lost matches), too. We work in two settings: in the \emph{untimed} one, a pattern is an NFA; in the \emph{timed} one, a pattern is a timed automaton. The extension of our untimed construction to the timed setting is technically involved, but our experiments demonstrate its practical benefits.
\end{abstract}


%
\IEEEpeerreviewmaketitle


\section{Introduction}\label{sec:intro}
\myparagraph{Monitoring and (Timed) Pattern Matching}
\IEEEPARstart{T}{he complexity} of \emph{cyber-physical systems (CPS)} has been rapidly growing,
due to increasingly  advanced digital control that realizes not only enhanced efficiency (e.g.\ in cars' fuel consumption) but also totally new functionalities such as automatic driving. Getting those systems right is, therefore, a problem that is as important, and as challenging, as ever. 

 Due to such complexity of CPS, combined with other reasons such as black-box components provided by other suppliers,  \emph{formal verification} in the conventional sense  is often hard to apply to real-world CPS. This has made researchers and practitioners turn to so-called \emph{lightweight formal methods}. \emph{Runtime verification} is a major branch therein, where  execution traces of a given system  is checked against a given specification~\cite{DBLP:conf/rv/2017cubes}. Various algorithms for \emph{monitoring} have been proposed for this purpose. 

Mathematically speaking, one common formalization of the monitoring  problem is  as the  problem of \emph{pattern matching}.\footnote{Another common formalization is as what we call the  \emph{pattern search} problem. Pattern search is easier than pattern matching, but is less informative. 
Their comparison is at the end of~\S{}\ref{sec:intro}.
} In case an execution trace is given by a word $w=a_{1}a_{2}\dotsc a_{n}$,
the expected outcome  is the set 
\begin{equation}\label{eq:patternMatchingOutcomeUntimed}
\begin{array}{r}
 \Match(w,\mathsf{pat})
 \;:=\;
 \{(i,j)\mid w|_{[i,j]}
 \models \mathsf{pat}\}
  \\
  \qquad\text{(where $w|_{[i,j]}=a_{i}a_{i+1}\dotsc a_{j}$)}
\end{array}
\end{equation}
of pairs $(i,j)$ of indices, the restriction of $w$ to which satisfies the given pattern 
$\mathsf{pat}$. The pattern $\mathsf{pat}$ can be given by a string, a set of strings, a regular expression, an automaton, etc.

\begin{myexample}\label{ex:untimedPatMatch}
Consider a word $w_{1}=\mathtt{abb\,bbb\,aab}$  and a pattern given by a regular expression $\mathcal{A}_{1}=\mathtt{aa}^{*}b$. There are three matches, and we have 
 $\Match(w_{1},\mathcal{A}_{1})=\{(1,2),(7,9),(8,9)\}$.
\end{myexample}

 More interesting in the CPS context is the \emph{timed} version of pattern matching. In one common formalization, an execution trace is given by a \emph{timed word}---a sequence of time-stamped characters such as $w_{2}=(\mathtt{a}, 0.1) (\mathtt{b}, 2.5) (\mathtt{a}, 3.5) (\mathtt{b}, 4.8)$. A pattern  $\mathsf{pat}$ is then specified by a \emph{timed automaton (TA)} $\mathcal{A}$ from~\cite{Alur1994}, and we compute the set
\begin{equation}\label{eq:patternMatchingOutcomeTimed}
 \Match(w,\mathcal{A})
 \;:=\;
  \bigl\{\,(t,t')\in \Rnn^{2}\;\big|\; t< t'\text{ and } w|_{(t,t')}\in L(\mathcal{A})\,\bigr\}
\end{equation}
of intervals $(t,t')$, the restriction of $w$ to which is accepted by the TA $\mathcal{A}$. Unlike the untimed setting, a TA $\mathcal{A}$ allows one to express various real-time constraints, which leads to a much more refined analysis of execution traces of CPS. 

\begin{myexample}\label{ex:timedPatMatch}
Consider the  timed word $w_{2}=(\mathtt{a}, 0.1) (\mathtt{b}, 2.5) (\mathtt{a}, 3.5) (\mathtt{b}, 4.8)$, and the pattern ``$\mathtt{b}$ occurs within two seconds after $\mathtt{a}$ does'' (a TA for a pattern that is essentially the same  is in Fig.~\ref{fig:case2_pattern}). Any match contains the second succession of $\mathtt{a}$ and $\mathtt{b}$---note that the first succession is too far apart. One such match is given by $w_{2}|_{(3,5)}$; but there are uncountably many such matches. The match set can be expressed symbolically as $\{(t,t')\mid 2.5\le t< 3.5, 4.8<t'\}$. 
\end{myexample}

Despite its obvious applications in various stages of CPS design and
deployment, the study of  timed pattern matching  was started only
recently~\cite{DBLP:conf/formats/UlusFAM14,DBLP:conf/tacas/UlusFAM16,DBLP:conf/cav/Ulus17,DBLP:conf/formats/WagaHS17,DBLP:conf/formats/BakhirkinFMU17,DBLP:conf/formats/AsarinMNU17}. Consequently,
the use of timed pattern matching is quite limited in the current industry practice: for example, the existence of algorithms that can match against temporal specifications is not  commonly known.

\begin{figure*}[tbp]
\begin{minipage}{.78\textwidth}
  \centering
 \includegraphics[width=.9\linewidth]{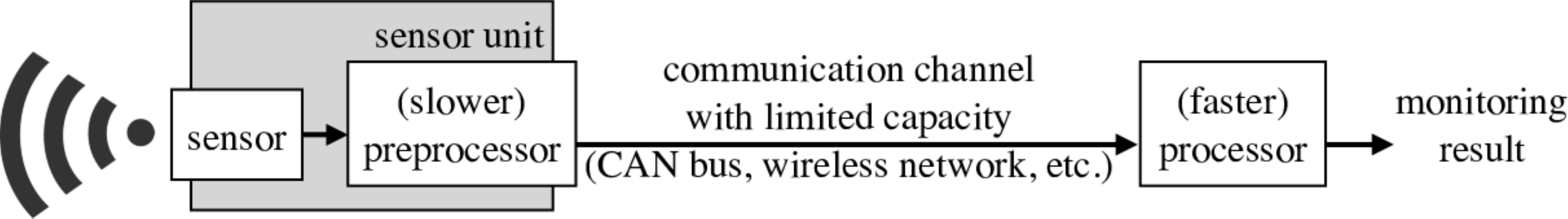}
 \caption{Common hardware architecture for embedded monitoring}
 \label{fig:hardwareArch}
 \scalebox{.8}{
 \parbox{\linewidth}{
 \begin{displaymath}
 \xymatrix@1@C+5em{
 {}
 \ar[r]^-{\text{exec.\ trace}}_-{\txt{(a word, or\\ a timed word)}}
 &
 *+[F]{\text{filter $\mathcal{M}_{N,\mathcal{A}}$}}
  \ar[r]^-{\txt{filtered word}}
 &
  *+[F]{\txt{pattern matching
 \\ against pattern $\mathcal{A}$}}
 \ar[r]^-{\txt{matching\\ result}}_-{\txt{(intervals)}}
 &
 {}
 }
 \end{displaymath}}}
 \caption{Proposed monitoring workflow with filtering}
 \label{fig:workflow}
\end{minipage}
\begin{minipage}{.2\textwidth}
\centering
\scalebox{0.6}{\begin{tikzpicture}
 \draw (0,3) rectangle (5,2.5);
 \node at (2.5,2.75) {$w$};

 \foreach \x / \y in {1.0/2.2,1.5/1.9
}
 {
 \draw [thick,|-|](0,\y)--(\x,\y);
 \node at (0.5,\y) {$\blacktriangleright$};
 }
 \node at (1.0,1.6) {$\vdots$};
 \draw [thick,|-|](0,1.2)--(5,1.2);
 \node at (0.5,1.2) {$\blacktriangleright$};
 \foreach \x / \y in {1.5/0.7,2.0/0.4
}
 {
 \draw [thick,|-|](0.5,\y)--(\x,\y);
 \node at (1.0,\y) {$\blacktriangleright$};
 }
 \node at (1.25,0.1) {$\vdots$};
 \draw [thick,|-|](0.5,-0.3)--(5,-0.3);
 \node at (1.0,-0.3) {$\blacktriangleright$};

 \node at (2.5,-0.7) {${\Large \vdots}$};

 \foreach \x / \y in {3.5/-1.3,4.0/-1.6
}
 {
 \draw [thick,|-|](2.5,\y)--(\x,\y);
 \node at (3.0,\y) {$\blacktriangleright$};
 }
 \node at (3.5,-1.9) {$\vdots$};
 \draw [thick,|-|](2.5,-2.3)--(5.0,-2.3);
 \node at (3.0,-2.3) {$\blacktriangleright$};
\end{tikzpicture}}
\caption{Back and forth in pattern matching}
\label{fig:backAndForth}
\end{minipage}
\end{figure*}


\myparagraph{Remote Monitoring in Embedded Applications}
In this paper, we propose \emph{filtering} for (timed and untimed) pattern matching. It is preprocessing that is applied to an input word.

Our motivation  comes from \emph{embedded} applications. In embedded systems---an important aspect of CPS---it is common that a sensor unit and a processor (that conducts the computational task of monitoring) are placed  physically apart. Moreover,  the communication channel between them often has a limited capacity. See Fig.~\ref{fig:hardwareArch}. 

An example of such a situation is in a modern automobile, in which a sensor unit in the engine gathers data and sends them to a processor that is far apart (to avoid the engine's heat and vibration, for example). They are interconnected by a \emph{controller area network (CAN)}, which is subject to severe performance limitation due to cost reduction. Another example is in an \emph{IoT device} like a connected electronic appliance and a connected car. It continually sends its own status to a server, that is set up in the computer cloud and monitors the device. The wireless communication channel between them is  limited e.g.\ by the battery capacity of the device.

\myparagraph{Filtering for (Timed) Pattern Matching}
A natural idea in such  remote monitoring situations is to try to reduce the amount of data that is sent from the sensor unit to the processor, in such a way that does not affect the result of monitoring. Many sensor units come with processors within, and we can use them for such preprocessing. 
We assume that the preprocessor (within the sensor unit, Fig.~\ref{fig:hardwareArch}) is much slower than the processor that conducts actual monitoring; therefore the preprocessing must be computationally cheap.

This leads to our proposed workflow shown in Fig.~\ref{fig:workflow}. There we  apply computationally cheap \emph{filtering} to the input word, in order to reduce the load on the communication channel to the processor, and to reduce the load on the processor as well.

\myparagraph{Moore Machines as Filters} 
In this paper we focus on two settings of monitoring: 
(1) in
  the \emph{untimed} setting, an execution trace is a word $w\in\Sigma^{*}$ and a pattern is a nondeterministic finite automaton (NFA) $\mathcal{A}$ over $\Sigma$; and
(2) in
 the \emph{timed} setting,  an execution trace is a timed word and a pattern is a timed automaton (TA) $\mathcal{A}$. 
Our technical contribution is the construction of a filter $\mathcal{M}_{\mathcal{A},N}$ realized as a \emph{Moore machine}, based on a pattern automaton $\mathcal{A}$ and a parameter $N\in\N$ called the \emph{buffer size}. Moore machines are a well-known model of state-based computation: it is an automaton with an additional state-dependent \emph{output} function. It operates in a nicely sequential and synchronous manner, reading one input character,  moving to a next state and outputting one character. This feature  is especially suited to
the \emph{logic synthesis} of digital circuits~\cite{Reese06}, opening up the way to hardware acceleration by FPGA or ASIC. 

Such sequential operation of Moore machines is in stark contrast with that of pattern matching. Since a pattern is given by an automaton $\mathcal{A}$ in our settings, the length of matches (i.e.\ $|w|_{[i,j]}|=j-i+1$ such that $w|_{[i,j]}\in L(\mathcal{A})$) is not fixed. Therefore we have to try matching windows of different size at different positions, going back and forth over the input word $w$  (see Fig.~\ref{fig:backAndForth}). This exhibits a qualitative difference between the (sequential) filtering task conducted by the (slower) preprocessor, and the (back-and-forth) pattern matching task conducted by the (faster) main processor. Our construction yields an  (untimed) Moore machine as a filter $\mathcal{M}_{\mathcal{A},N}$, even in the timed setting.

The output of our filter $\mathcal{M}_{\mathcal{A},N}$ is the same as input (timed) word $w$, except that some characters
are \emph{masked} by a fresh character $\bot$. For example: $w=\mathtt{abb\,bbb\,aab}$  is turned into $\mathtt{ab\bot\,\bot\bot\bot\, aa b}$ under the pattern $\mathcal{A}=\mathtt{aa}^{*}b$. By the binary representation of the length of successive $\bot$'s, the data size can be reduced exponentially. Moreover, in case we are only interested in the matched subwords $w|_{[i,j]}$ (but not in the indices $i,j$), we can further suppress successive $\bot$'s into one $\bot$. (Note that removing all $\bot$'s after the filtering stage can result in spurious matches in the pattern matching stage, see Fig.~\ref{fig:workflow}.)

 Our Moore-machine filter $\mathcal{M}_{\mathcal{A},N}$ is constructed based on a pattern automaton $\mathcal{A}$ and a positive integer $N$ for the \emph{buffer size}. The parameter $N$ allows a user to choose the balance between the computational cost of filtering (the greater $N$ is, the more states $\mathcal{M}_{\mathcal{A},N}$ has) and the size of the filtered word (the greater $N$ is, the more $\bot$, i.e.\ the smaller the filtered word is). This flexibility makes the algorithm suited for various hardware configurations. 

We have implemented our construction; we present our experiment results for the \emph{timed} setting (which is harder). Our examples come from the automotive domain. We observe that for realistic pattern TAs $\mathcal{A}$ and input timed words $w$, the filtered words can be 2--100 times shorter than the original word $w$. We also experimentally confirm that running a Moore machine  $\mathcal{M}_{\mathcal{A},N}$  is cheap. Furthermore, we observe that having a filter (like in Fig.~\ref{fig:workflow}) accelerates the  task of timed pattern matching itself, by 1.2--2 times. 

In the theoretical analysis of  our  construction, we prove \emph{soundness}: all the matches in the original input word are preserved by filtering. We show soundness
for both of the untimed and timed settings. 
We note, however, that  soundness is satisfied by the trivial (identity) filter; so soundness itself does not speak much about the benefit of filtering. Besides our experiments, 
we present some theoretical results about the filtering performance 
  in the untimed setting.  They include the following: completeness (in the sense that \emph{all} the unnecessary characters get masked) if $L(\mathcal{A})$ is finite (this is the setting of \emph{multiple string matching}~\cite{DBLP:journals/cacm/AhoC75}); 
and \emph{monotonicity} (bigger 
$N$  leads to better filtering results). These results also suggest performance advantages of our \emph{timed} construction since it shares the basic ideas with the untimed one.

Our construction of filters is automata-theoretic, with two principal steps of 1) equipping a buffer (of size $N$), and 2) determinization. For the second step in the timed setting, we employ \emph{one-clock determinization} of timed automata (TA)~\cite[\S{}5.3]{DBLP:journals/fmsd/KrichenT09} that overapproximates a given TA by a deterministic and one-clock TA.

\myparagraph{Contribution} Overall, our contribution is summarized as follows. 
 \begin{itemize}
  \item A construction of a filter $\mathcal{M}_{\mathcal{A},N}$ for
        \emph{untimed} pattern matching against $\mathcal{A}$. The
        filter is given by a Moore machine and thus operates in a
        simple, sequential and synchronous way. It is also amenable to
        hardware acceleration by logic circuits. The parameter $N$ gives
        flexibility to a user  in the trade-off between computational cost and the effect of filtering. 
  \item A construction of a filter $\mathcal{M}_{\mathcal{A},N}$ for \emph{timed} pattern matching. Given a timed automaton $\mathcal{A}$ as a pattern and the buffer size, we construct an (untimed) Moore machine $\mathcal{M}_{\mathcal{A},N}$. The construction extends the untimed version, and is more involved, employing zone-based abstraction of the pattern timed automaton. Given also its practical relevance, we consider the timed construction to be the main contribution of the paper. 

  \item Soundness (preservation of all matches) is proved in the untimed and timed settings. 
  \item In the timed setting we also prove some theoretical results about the performance of our filters.
  \item Implementation of the timed construction, and experiments that demonstrate the benefit of our filtering construction.
\end{itemize}

\myparagraph{Pattern Matching vs.\ Pattern Search} 
Another mathematical formulation of monitoring---that is alternative to our choice of pattern matching---is what we call the \emph{pattern search} problem. It asks if the match set (see~(\ref{eq:patternMatchingOutcomeUntimed}--\ref{eq:patternMatchingOutcomeTimed}))
is empty or not.  Pattern search is appealing since it is easily reduced to the \emph{membership problem}. Roughly speaking, given a pattern automaton $\mathcal{A}$, one first adds a self-loop to the initial state so that prefixes of an input word can be disregarded, and then monitors if any accepting state becomes active.  Pattern search has been extensively studied in the monitoring context: see e.g.~\cite{DBLP:conf/uss/MeinersPNTL10,DBLP:journals/tosem/BauerLS11}.

 Pattern \emph{matching}  is more expensive than pattern search, since it requires remembering  indices (Fig.~\ref{fig:backAndForth}). It is nevertheless highly relevant to real-world monitoring applications---especially to \emph{remote} monitoring such as in Example~\ref{ex:introRemoteMonitoring} below. Note that remote monitoring is often only \emph{semi}-online:  a log can arrive at a monitor in sizable chunks, in an intermittent manner. 
Then it is imperative to be able to  single out which part of the received log is issuing an alert. 

The relevance of pattern matching to monitoring applications is widely acknowledged in the community, as is witnessed by the recent proliferation of  literature~\cite{DBLP:conf/formats/WagaAH16,DBLP:conf/formats/WagaHS17,DBLP:conf/hybrid/UlusM18,DBLP:conf/cav/Ulus17,DBLP:conf/tacas/UlusFAM16,DBLP:conf/cav/FerrereMNU15,DBLP:conf/formats/UlusFAM14,KandhanTP10,DBLP:conf/sigmod/MajumderRV08,DBLP:conf/ancs/YuCDLK06}.

\begin{myexample}[semi-online remote monitoring]\label{ex:introRemoteMonitoring}
 As a concrete example of remote monitoring (Fig.~\ref{fig:hardwareArch}), let us imagine a semi-connected vehicle. It keeps driving logs in its memory, and sends them to a center over the Internet once it stops within the range of a known wireless network. Analysis of the logs is done in the center.
A driving log is a timed word with information on the vehicle's position, velocity and throttle. One such timed word $w$, taken from ROSBAG
STORE (\url{rosbag.tier4.jp}), looks like the left below after we plot the vehicle's position. (The plot is discontinuous: the absence of data can be attributed e.g.\ to loss of the GPS signals.)


\begin{figure}[h]
 \begin{minipage}{0.23\textwidth}
 \centering
 \scalebox{0.27}{
  \includegraphics{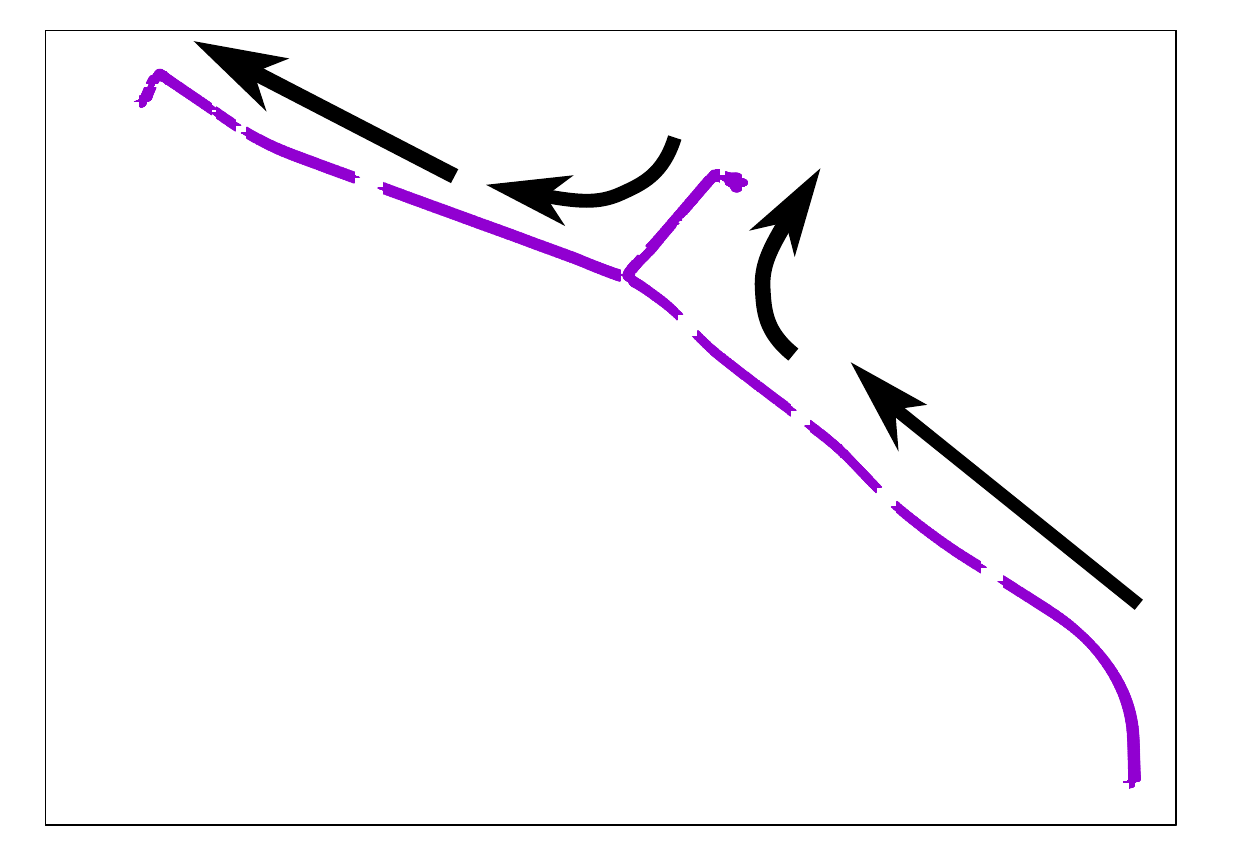}
 }
 \end{minipage}
 \begin{minipage}{0.23\textwidth}\centering
 \centering
 \scalebox{0.27}{
  \includegraphics{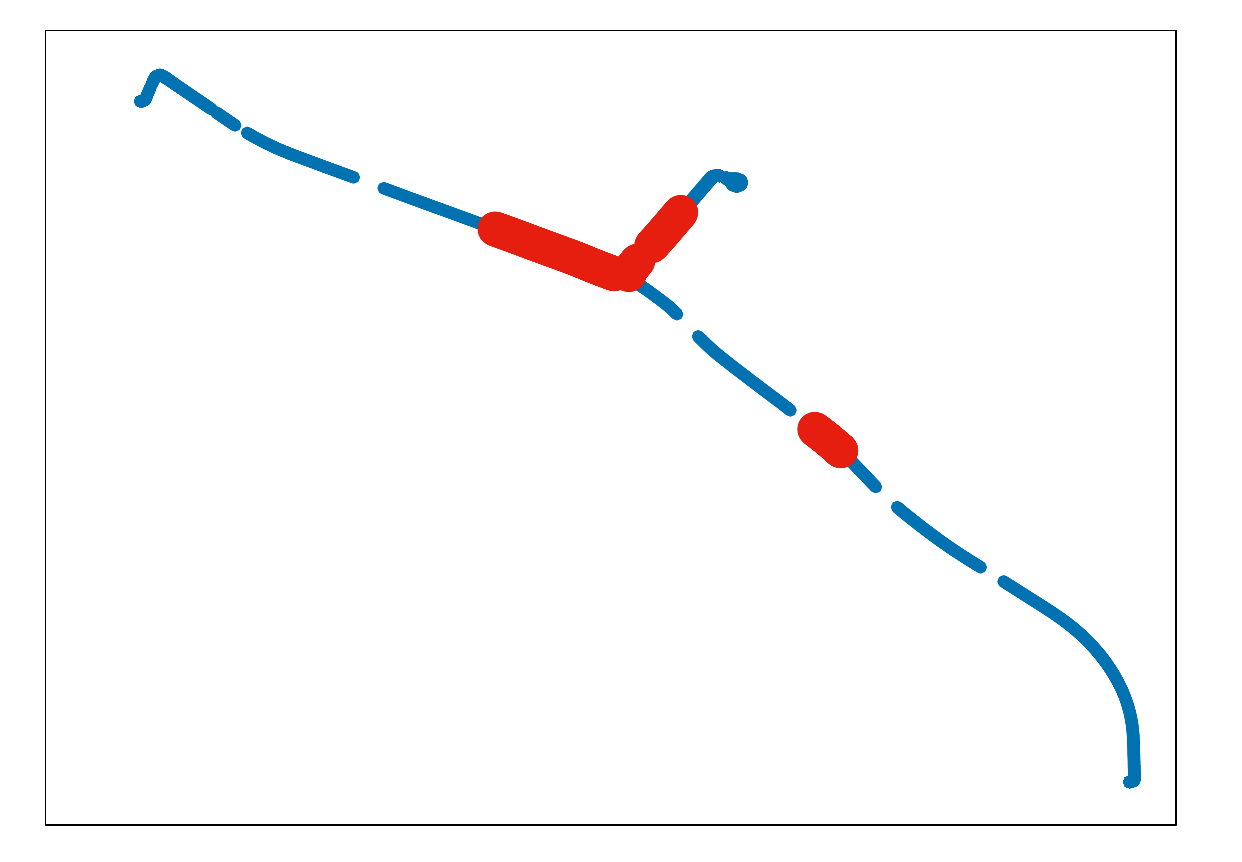}
 }
\end{minipage}
\end{figure}

Let us say that we are interested in those road segments where the throttle
 is greater than a certain threshold for  ten seconds or longer. We run
 timed pattern matching on $w$ with a suitable pattern automaton
 $\mathcal{A}$. By mapping the identified time intervals to the position
 plot, we identify the road segments of our interests (right above, red, generated by the tool \monaa~\cite{MONAACode}). 

\end{myexample}



\myparagraph{Organization} 
We fix  notations in~\S{}\ref{sec:prelim}. We present a construction of filter Moore machines for \emph{untimed} pattern matching, in~\S{}\ref{sec:untimed}. We also prove some properties such as soundness. The same idea is used in~\S{}\ref{sec:timed} for a more complicated problem of filtering for \emph{timed} pattern matching. Here we prove soundness. Implementation and experiment results for the timed case are presented in~\S{}\ref{sec:experiments}. We discuss related work in~\S{}\ref{sec:related}.
Most of the proofs are deferred to the appendix.

\section{Preliminaries}\label{sec:prelim}

The set $\Sigma^{*}=\bigcup_{n\in\N}\Sigma^{n}$ is that of \emph{words} over $\Sigma$. The length $n$ of a word $w=a_{1}a_{2}\dotsc a_{n}$ (where $a_{i}\in \Sigma$) is denoted by $|w|$. 

For an NFA $\mathcal{A} = (\Sigma, S, s_0, S_F, E)$ and a string 
$w\in\Sigma^*$ over the common alphabet $\Sigma$, a \emph{run}
$\overline{s}$ of $\mathcal{A}$ over $w$ is a sequence 
$\overline{s} = s_0,s_1,\dots,s_{|w|}$ such that for each $i \in[1,|w|]$
we have $(s_{i-1},w_i,s_{i})\in E$.
A run $\overline{s} = s_0,s_1,\dots,s_{|w|}$ is \emph{accepting} if we
have $s_{|w|} \in S_F$.

The powerset of a set $X$ is denoted by $\mathcal{P}(X)$. The disjoint union of $X,Y$ is $X\amalg Y$. 
For an alphabet $\Sigma$,  the set $\Sigma \amalg \{\bot\}$ augmented with a fresh symbol $\bot$ is denoted by $\Sigma_{\bot}$.

We will be using the set $\{1,2,\dotsc, N\}$ for the value domain of counters. This is denoted by $\mathbb{Z}/N\mathbb{Z}$, because we rely on its algebraic structure (such as addition modulo $N$).

 A \emph{Moore machine} 
is given by
 $\mathcal{M}=(\Sigma_{\mathrm{in}},\Sigma_{\mathrm{out}},Q,q_0,\Delta,\Lambda)$
 where $\Sigma_{\mathrm{in}}$ and $\Sigma_{\mathrm{out}}$ are  input
 and output alphabets, $Q$ is a finite set of states, $q_0 \in Q$ is an
 initial state, $\Delta: Q \times \Sigma_{\mathrm{in}}\to Q$ is a transition
 function, and
 $\Lambda: Q \to\Sigma_{\mathrm{out}}$ is an
 output function.
 For
 $\mathcal{M}$
 and an input word
 $w = a_1 a_2\dotsc a_n \in \Sigma_{\mathrm{in}}^*$ (where $a_{i}\in\Sigma_{\mathrm{in}}$),
 the \emph{run} $\overline{q}$ of $\mathcal{M}$ over $w$ 
 is the sequence $\overline{q} = q_0 q_1 \dotsc q_n \in Q^*$ satisfying
 $q_i=\Delta(q_{i-1},a_{i})$ for each $i \in [1,|w|]$.
 In this case,
 the \emph{output word} $w'\in\Sigma_{\mathrm{out}}^*$ of
 $\mathcal{M}$ over 
 $w$ is given by
 $w'=\Lambda(q_0)\Lambda(q_{1})\dotsc\Lambda(q_{n-1})\in\Sigma_{\mathrm{out}}^*$.


\section{Moore-Machine Filtering for Pattern Matching I: Untimed}
\label{sec:untimed}

\subsection{Problem Formulation}

\begin{definition}[(untimed) pattern matching]\label{def:patternMatching}
 \label{def:pattern_matching}
 Given an NFA $\mathcal{A}$ over an alphabet $\Sigma$ and a word
 $w=a_{1}a_{2}\dotsc a_{n}\in\Sigma^{*}$, the \emph{pattern
 matching} problem asks for the \emph{match set}
\begin{math}
 \Match(w,\mathcal{A})= \bigl\{\,(i,j)\in\N^{2}\;\big|\; w|_{[i,j]}\in L(\mathcal{A})
\,\bigr\}
\end{math}, where
 $w|_{[i,j]}=a_{i}a_{i+1}\dotsc a_{j}$.
\end{definition}

Our goal is the workflow shown in Fig.~\ref{fig:workflow}. We fix the input/output type of the filter in the following general definition.
\begin{definition}[Moore-machine filter for (untimed) pattern matching]
 \label{def:filterUntimed}
 Let $\mathcal{A}$ be an NFA over an alphabet $\Sigma$, and let $N$ be a positive integer. A \emph{filter for $\mathcal{A}$ with buffer size $N$} is a Moore machine $\mathcal{M}=(\Sigma_{\mathrm{in}},\Sigma_{\mathrm{out}},Q,q_0,\Delta,\Lambda)$ that satisfies the following. 
\begin{itemize}
 \item $\Sigma_{\mathrm{in}} = \Sigma_{\mathrm{out}} =\Sigma_{\bot}$.
 \item Let $w=a_{1}\dotsc a_{n}\in \Sigma^{*}$ be an arbitrary word, and consider the word $w\bot^{N}$ obtained by padding $\bot$'s in the end. We require that the output word of $\mathcal{M}$ over this word $w\bot^{N}$  be of the form $\bot^{N}w'$, where $w'=b_{1}\dotsc b_{n}$, and  $b_{i}$ is either $\bot$ or $a_{i}$ for each $i\in [1,n]$. We say that the character $a_{i}$  at the position $i$ is \emph{passed} if $b_{i}=a_{i}$; otherwise (i.e.\ if $b_{i}=\bot$) we say $a_{i}$ is \emph{masked}. 
\end{itemize}
 A filter $\mathcal{M}$ is \emph{sound} if it preserves all matching intervals. That is,
$b_{k}=a_{k}$
for each $k\in[1,n]$ such that 
$\exists i,j.\, (k\in [i,j] \land [i,j]\in \Match(w,\mathcal{A}))$.
\end{definition}

\begin{wrapfigure}[9]{r}{0pt}
 \includegraphics[clip,width=.50\linewidth]{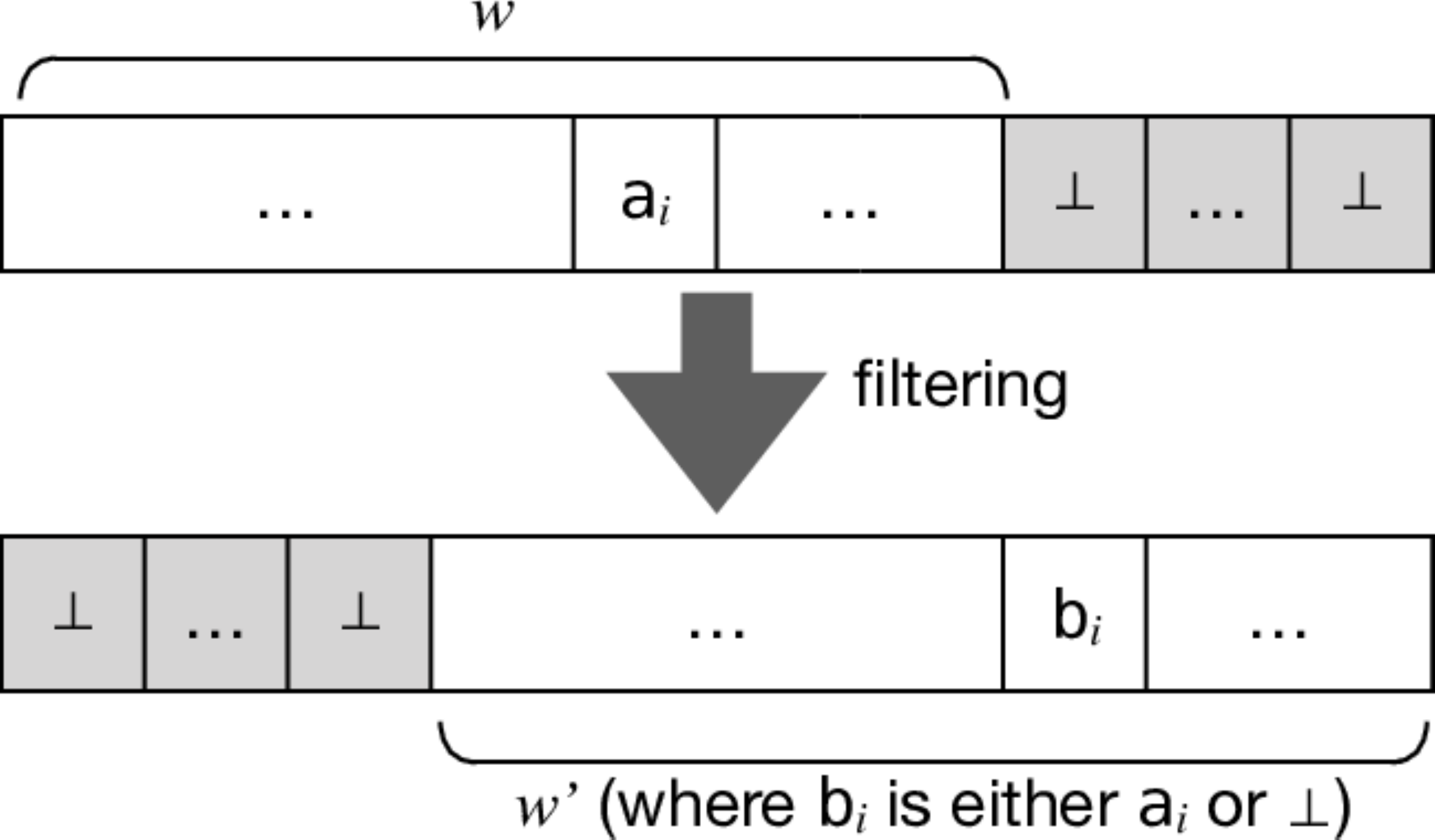}
  \caption{Padding in filtering}
  \label{fig:filteringPadding}
\end{wrapfigure}
Explanations are in order, especially about the buffer size $N$ and the padding by $\bot^{N}$ of the input and output words. The padding means filtering is done in the way depicted in Fig.~\ref{fig:filteringPadding}, with a delay of $N$ steps. This is because of how our Moore machine works (Fig.~\ref{fig:filteringSteps}): it  scans the input word from left to right; the machine stores $N$ characters in the FIFO buffer in it (encoded in its state space $Q$); and the characters are output once they are dequeued from the FIFO buffer. Hence there is a delay of $N$ steps. Initially, the buffer is filled with $\bot$ (Fig.~\ref{fig:filteringSteps}, Step 0); this accounts for the prefix $\bot^{N}$ of the output word $\bot^{N}w'$. The  padding $\bot^{N}$ at the end of the input word $w\bot^{N}$ is needed to dequeue the content of the buffer (from Step $n+1$ to $n+N$). 
In the course of the scanning in Fig.~\ref{fig:filteringSteps}, some characters in $w=a_{1}\dotsc a_{N}$ get masked, although such masking is not explicit in the figure.




\begin{figure*}[tbp]
 \centering
 \includegraphics[width=.9\linewidth]{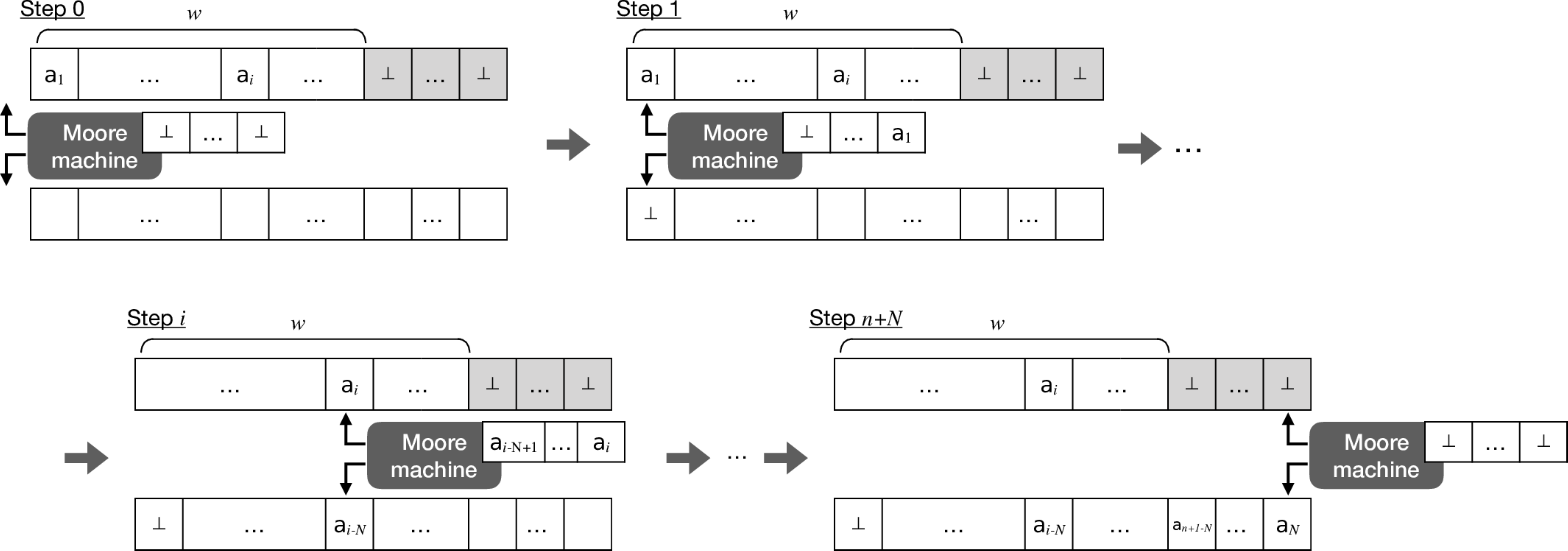}
 \caption{Filtering by a Moore machine. Sometimes a character $a_{i}$ is decided to be unnecessary and gets masked by $\bot$ (i.e.\ $b_{i}=\bot$), although such masking is not explicit in the figure. } 
 \label{fig:filteringSteps}
\noindent
\begin{minipage}{.5\textwidth}
   \centering

  \scalebox{.6}{
  \begin{tikzpicture}[shorten >=1pt,node distance=2.0cm,on grid,auto]
  \clip(-1.8,-0.5) rectangle (4.5,1.2);
  \node[initial,state] (S0){$s_0$};
  \node[state] (S1)[right= of S0] {$s_1$};
  \node[accepting,state] (S2) [right =of S1]{$s_2$};

  \path[->]
  (S0) edge[above] node {$\tt a$} (S1)
  (S1) edge[above] node {$\tt b$} (S2)
  (S1) edge[loop above] node {$\tt a$} (S1);
  \end{tikzpicture} 
 }
 \hspace{-7em}
  \scalebox{0.6}{
  \begin{tikzpicture}[shorten >=1pt,node distance=2.5cm,on grid,auto]
  \node[initial,state,label=above right:$q'_{0}$] (S0){$(s_0,0)$};
  \node[state,label=below right:$q'_{3}$] (S1) [below right=of S0] {$(s_0,0),(s_2,2)$};
  \node[state,label=above right:$q'_{1}$] (S2) [above right=of S1] {$(s_0,0),(s_1,1)$};
  \node[state,inner sep=0pt, minimum size=0pt,label=above left:$q'_{2}$] (S4) [right=3cm of S2] {\begin{tabular}{c}
                                        $(s_0,0),(s_1,1),$\\
                                        $(s_1,2)$
                                       \end{tabular}};
  \node[state,inner sep=0pt, minimum size=0pt,label=left:$q'_{4}$] (S3) [below=3cm of S4] {\begin{tabular}{c}
                                        $(s_0,0),(s_2,1),$\\
                                        $(s_2,2)$
                                       \end{tabular}};

  \path[->]
  (S0) edge[loop above] node {$\tt b$} (S0)
  (S4) edge[loop right] node {$\tt a$} (S4)
  (S0) edge[above] node {$\tt a$} (S2)
  (S1) edge[below] node {$\tt b$} (S0)
  (S2) edge[above] node {$\tt a$} (S4)
  (S3) edge[below] node {$\tt a$} (S2)
  (S4) edge node {$\tt b$} (S3)
  (S2) [bend left = 30]edge[below] node {$\tt b$} (S1)
  (S1) [bend left = 30]edge[left] node {$\tt a$} (S2)
  (S3) [bend left = 55]edge[below left] node {$\tt b$} (S0);
  \end{tikzpicture} 
 }
 \caption{Left: a pattern NFA $\mathcal{A}_{0}$ for $\mathtt{aa}^{*}\mathtt{b}$. 
\,
Right:
the filter Moore machine $\mathcal{M}_{\mathcal{A}_{0},2}$, its non-buffer part}
 \label{fig:filterMooreMachineExample}
\end{minipage}
\hfill
\begin{minipage}{.48\textwidth}
\noindent
\scalebox{.8}{ 
\noindent
\begin{minipage}{\linewidth}
 \begin{displaymath}
\begin{array}{l}
 \Bigl(
  q'_{0},\,
  \begin{array}{|c|c|}\hline
       \bot &\bot \\\hline \mask &\mask
\\\hline  \end{array}
   \,\Bigr)
 \xrightarrow[\bot]{\tt a}
 \Bigl(
  q'_{1},\,
  \begin{array}{|c|c|}\hline
       \bot & \tt a \\\hline \mask &\mask
\\\hline  \end{array}
   \,\Bigr)
 \\
 \xrightarrow[\bot]{\tt b}
 \Bigl(
  q'_{3},\,
  \begin{array}{|c|c|}\hline
       \tt a & \tt b \\\hline \pass &\pass
\\\hline  \end{array}
   \,\Bigr)
 \xrightarrow[\tt a]{\tt b}
 \Bigl(
  q'_{0},\,
  \begin{array}{|c|c|}\hline
       \tt b & \tt b \\\hline \pass &\mask
\\\hline  \end{array}
   \,\Bigr)
 \\
 \xrightarrow[\tt b]{\tt b}
 \Bigl(
  q'_{0},\,
  \begin{array}{|c|c|}\hline
       \tt b & \tt b \\\hline \mask &\mask
\\\hline  \end{array}
   \,\Bigr)
 \xrightarrow[\bot]{\tt a}
 \Bigl(
  q'_{1},\,
  \begin{array}{|c|c|}\hline
       \tt b & \tt a \\\hline \mask &\mask
\\\hline  \end{array}
   \,\Bigr)
 \\
 \xrightarrow[\bot]{\tt a}
 \Bigl(
  q'_{2},\,
  \begin{array}{|c|c|}\hline
       \tt a & \tt a \\\hline \pass &\pass
\\\hline  \end{array}
   \,\Bigr)
 \xrightarrow[\tt a]{\tt b}
 \Bigl(
  q'_{4},\,
  \begin{array}{|c|c|}\hline
       \tt a & \tt b \\\hline \pass &\pass
\\\hline  \end{array}
   \,\Bigr)
 \\
 \xrightarrow[\tt a]{\bot}
 \Bigl(
  \{(s_0,0)\},\,
  \begin{array}{|c|c|}\hline
       \tt b & \bot \\\hline \pass &\mask
\\\hline  \end{array}
   \,\Bigr)
 \xrightarrow[\tt b]{\bot}
 \Bigl(
  \{(s_{0},0)\},\,
  \begin{array}{|c|c|}\hline
       \bot & \bot \\\hline \mask &\mask
\\\hline  \end{array}
   \,\Bigr)
\end{array}
\end{displaymath}\end{minipage}
} \caption{The run of
 $\mathcal{M}_{\mathcal{A}_{0},2}$ over the word $w=\mathtt{abbbaab}$. The tables show the states of the buffer, enqueued from the right.
In $\xrightarrow[b]{a}$, $a$ is the input character and $b$ is the output character. }
 \label{fig:exampleRunOfFilterMooreMachine}
\end{minipage}
\end{figure*}




\subsection{Construction of  Filter Moore Machine $\mathcal{M}_{\mathcal{A},N}$}
\label{subsec:filterConstructionUntimed}

\begin{definition}[the filter $\mathcal{M}_{\mathcal{A},N}$ for (untimed) pattern matching]
\label{def:filterConstructionUntimed}
 Let $\Sigma$ be an alphabet, $N$ be a positive integer, and 
 $\mathcal{A} = (\Sigma, S, s_0, S_F, E)$ be an NFA. 
 We define a Moore machine 
\begin{displaymath}
  \mathcal{M}_{\mathcal{A},N} = (\Sigma_{\bot},
 \Sigma_{\bot},Q,q_0,\Delta,\Lambda)
\end{displaymath} 
as follows. Its input and output alphabets are both $\Sigma_{\bot}$. 

 The state space is  $Q=\mathcal{P}\bigl(\,S\times (\mathbb{Z}/N\mathbb{Z})\,\bigr)\times\bigl((\Sigma_{\bot})^{N}\times\{\pass,\mask\}^{N}\bigr)$. Here $\mathbb{Z}/N\mathbb{Z}$ is the $N$-element set with addition modulo $N$ (\S{}\ref{sec:prelim}). 

 The initial state is $q_0 = \bigl(\,\bigl\{(s_0,0)\bigr\},\,(\bot,\dotsc,\bot),\,(\mask,\dotsc,\mask)\,\bigr)$. 

 The transition $\Delta\colon Q\times \Sigma_{\bot} \to Q$ is  as follows. For each $a\in\Sigma_{\bot}$,
       \begin{align}
&
\begin{array}{l}
 	\Delta\Bigl(\,\bigl(\mathcal{S},(a_1,a_2,\dots,a_{N}),(l_1,l_2,\dots,l_{N})\bigr),\,a\,\Bigr) 
	\\
	\qquad\qquad\qquad =
	\bigl(\,\mathcal{S}',\,(a_2,\dots,a_{N},a),\,\overline{l}\,\bigr)
	\enspace, \quad
 \text{ where}\quad
\end{array}	
\label{eq:constFilterUntimed1}
        \\
	&
\begin{aligned}
\mathcal{S}'=\bigl\{\,(s',(n \bmod N)+1)\mid (s,n) &\in \mathcal{S},(s,a,s')\in E\,\bigr\}\\
 &\cup\{(s_0,0)\}\enspace,\text{ and}
\end{aligned}
	\label{eq:constFilterUntimed2}
	\\
	&
  \overline{l} = 
  \footnotesize
 \begin{cases}
  (\pass,\dotsc,\pass)&\text{if $\exists s.\, (s,N)\in\mathcal{S}'$,}\\
  (l_2,l_3,\dots,l_{N-\psi(\mathcal{S}')+1},\overbrace{\pass,\dotsc,\pass}^{\psi(\mathcal{S}')})&
  \text{else if  $\exists n, s\in S_F.\,$
  }
  \\&
  \text{ $(s,n)\in\mathcal{S}'$,}
  \\
  (l_2,l_3,\dots,l_{N},\mask)&\text{otherwise.}
 \end{cases}
	\label{eq:constFilterUntimed3}
       \end{align}
 Here  $\psi(\mathcal{S}')\in \Zp$ is
 \begin{math}
  \psi(\mathcal{S}') = \max\{n\mid \exists s\in S_F.\,  (s,n)\in\mathcal{S}'\}
 \end{math}.


 Finally, the output function $\Lambda\colon Q\to \Sigma_{\bot}$ is defined as follows. 
\begin{equation}
 \label{eq:constFilterUntimed5}
 \Lambda\bigl(\,\mathcal{S},\,(a_1,a_2,\dotsc,a_{N}),\,(l_1,l_2,\dotsc,l_{N})\,\bigr)=
 \begin{cases}
  a_1& \text{if } l_1 = \pass\\
  \bot& \text{if }  l_1 = \mask
 \end{cases}
\end{equation}
\end{definition}

We describe  intuitions. The construction of $\mathcal{M}_{\mathcal{A},N}$ combines three  building blocks: \emph{determinization}, \emph{counters} and a \emph{buffer} of size $N$. 

\emph{(Determinization)} The  pattern $\mathcal{A}$ is an NFA, but we want a \emph{deterministic} Moore machine. This determinization accounts for the powerset construction $\mathcal{P}$ in the component $\mathcal{P}\bigl(\,S\times (\mathbb{Z}/N\mathbb{Z})\,\bigr)$ of the state space $Q$. For example, an element $\{\,(s_{1},n_{1}),\dotsc, (s_{k},n_{k})\,\}$ of this component means ``the states $s_{1},\dotsc, s_{k}$ are active in  the NFA $\mathcal{A}$.'' 
Observe that, in~(\ref{eq:constFilterUntimed2}), the part regarding states ($s,s',\dotsc$)  follows the usual determinization. An exception is the addition of $(s_{0},0)$ 
in~(\ref{eq:constFilterUntimed2}):
it is because a matching  can start at any position of the input word. 
 
\emph{(Counters)}
 Moreover, each active state that traverses $\mathcal{A}$ keeps a \emph{counter} for how many steps it has traveled since the initial state. This is the component $\mathbb{Z}/N\mathbb{Z}$ in the state space $Q$. The maximum value of those counters is the same as the buffer size $N$. Once this maximum is reached, a counter starts over from $1$. See~(\ref{eq:constFilterUntimed2}), where counters for active states are incremented modulo $N$. 

\emph{(Buffer)}
 A FIFO \emph{buffer} of size $N$ accounts for the second
       component $(\Sigma_{\bot})^{N}\times\{\pass,\mask\}^{N}$ of the
       state space $Q$. Each of the $N$ cells stores a character from
       $\Sigma_{\bot}$ and a label ($\pass$ or
       $\mask$). See~(\ref{eq:constFilterUntimed1}) and ~(\ref{eq:constFilterUntimed3})
, where the basic behavior of the buffer is to dequeue the leftmost element and to enqueue the read character on the right. 

It is the labels in the buffer ($\pass$ or $\mask$) that determine whether to mask a character or not. The default label is $\mask$ (the third case in~(\ref{eq:constFilterUntimed3})), and if it remains unchanged for $N$ steps then the corresponding character gets masked by $\bot$ in the output (the second case in~(\ref{eq:constFilterUntimed5})). 
The label can change from $\mask$ to $\pass$ for two different reasons (the first two cases in~(\ref{eq:constFilterUntimed3})). 
\begin{itemize}
 \item 
 The second case in~(\ref{eq:constFilterUntimed3}) is when it is detected that some characters towards the end of the buffer form a match against the pattern $\mathcal{A}$, leading to an accepting state $s\in S_{F}$ of $\mathcal{A}$. Then we mark those characters by $\pass$, indicating that they need to be passed to pattern matching (Fig.~\ref{fig:workflow}). The number $\psi(\mathcal{S}')$ of characters to be passed is calculated using the counter $n$ associated to the active state $s\in S_{F}$. 
 \item On the first case in~(\ref{eq:constFilterUntimed3}), its condition
$\exists s.\, (s,N)\in\mathcal{S}'$
       says that the counter for some active state $s$  has reached the maximum $N$. In this case we are not  sure whether this active state $s$ of $\mathcal{A}$ will eventually reach an accepting state or not. To be on the safe side,  we pass all the $N$ characters to pattern matching without masking them. In the untimed setting, this is the only place where completeness of filtering is potentially lost.
\end{itemize}
In summary, in Def.~\ref{def:filterConstructionUntimed} we construct a Moore machine that works in the way depicted in Fig.~\ref{fig:filteringSteps}. The Moore machine's state space combines the following: determinization of the pattern NFA $\mathcal{A}$; counters for the steps from the initial state; and a FIFO buffer that stores $N$ characters labeled by $\pass$ or $\mask$. 

\begin{proposition}\label{prop:sanityCheckUntimed}
 The Moore machine $\mathcal{M}_{\mathcal{A},N}$ is a filter for $\mathcal{A}$ with buffer size $N$, in the sense of Def.~\ref{def:filterUntimed}. 
\qed
\end{proposition}
\auxproof{ \begin{proof}
  Let
  $(\mathcal{S}_0,\overline{a}_0,\overline{l}_0),(\mathcal{S}_{1},\overline{a}_{1},\overline{l}_{1}),\dots,(\mathcal{S}_{|w|+N-1},\overline{a}_{|w|+N-1},\overline{l}_{|w|+N-1})\in Q^*$
  be the run of $\mathcal{M}_{\mathcal{A},N}$ over $w\cdot\bot^{N-1}$.
  Assume that there exists $i \in [1,N-1]$ satisfying $w'_{i}\neq\bot$.
  By the definition of $\Lambda$, we have ${l}_{i,1}=\pass$.
  Let $j$ be the index of $w$ satisfying
  $j=\min\{k \leq i \mid {l}_{k,(1+i-k)}=\pass\}$.
  Since we have $i < N$, there exists $(s,n)\in\mathcal{S}_{j}$ such that 
  we have $s\in S_F$ and $1+i-k+n > N$. Because of $n\leq k$, we have
  $i > N-1$, which contradicts the assumption.
  Thus, for any $i \in [1,N-1]$, we have $w'_{i} = \bot$.
  \end{proof}
}

 In our implementation, we realize a filter \emph{not}  as a plain Moore
 machine with a state space $Q=\mathcal{P}(S\times (\mathbb{Z}/N\mathbb{Z}))\times((\Sigma_{\bot})^{N}\times\{\pass,\mask\}^{N})$ as described in
 Def.~\ref{def:filterUntimed}.
 Instead, we separate the state space $Q$ into the ``buffer part''
 $(\Sigma_{\bot})^N \times \{\pass, \mask\}^N$ and the ``non-buffer part''
 $\mathcal{P}(S \times (\Z/N \Z))$, and we generate the former buffer part on-the-fly. More specifically, the non-buffer part is constructed as a DFA all at once in the beginning, and this DFA dictates how to operate on  the buffer part that is realized as an array of size $N$. See Example~\ref{ex:runUntimed} later for illustration. 

\begin{proposition}
 \label{prop:filter_state_space}
Let $\mathcal{A} = (\Sigma, S, s_0, S_F, E)$ be an NFA. 
For the induced filter Moore machine $\mathcal{M}_{\mathcal{A},N}$, the size of the 
non-buffer part 
 $\mathcal{P}(S\times(\Z/N \Z))$ of its state space is bounded by  $O(2^{N\cdot |S|})$.

 Therefore, the memory usage for the non-buffer part (including the transitions) is $O(2^{N\cdot|S|}\cdot |\Sigma|)$. The memory usage for the buffer part is $O(N\log|\Sigma|)$; overall, the space complexity of running our filter Moore machine $\mathcal{M}_{\mathcal{A},N}$ is $O(2^{N\cdot|S|}\cdot |\Sigma|)$. 
\qed
\end{proposition}
\noindent
The space complexity is exponential in $N$, and this comes from the powerset construction for the non-buffer part $\mathcal{P}(S\times(\Z/N \Z))$. Experimentally, however, the memory consumption does not necessarily grow exponentially in $N$. This is because not all states in 
$\mathcal{P}(S\times(\Z/N \Z))$ are reachable. See RQ2 in~\S{}\ref{sec:experiments}. 
\auxproof{ $O(2^{|S|\cdot N}\cdot (2|\Sigma|)^{N})$
}

\begin{example}\label{ex:runUntimed}
Let us consider the pattern $\mathtt{aa}^{*}\mathtt{b}$. It is expressed by the NFA $\mathcal{A}_{0}$ in Fig.~\ref{fig:filterMooreMachineExample}. The filter Moore machine 
$\mathcal{M}_{\mathcal{A}_{0},2}$ from Def.~\ref{def:filterConstructionUntimed} is depicted in Fig.~\ref{fig:filterMooreMachineExample}, where buffer states are omitted. Its run over the word  $w=\mathtt{abbbaab}$ is shown in Fig.~\ref{fig:exampleRunOfFilterMooreMachine}; its output word is $\bot\bot\mathtt{ab\bot\bot aa b}$, which means that the filtering result is
 $\mathtt{ab\bot\bot aa b}$. 
\end{example}

\subsection{Properties of the Moore Machine $\mathcal{M}_{\mathcal{A},N}$}
\label{subsec:filterPropertiesUntimed}
Throughout the rest of this section, 
 let $\mathcal{A}$ be a pattern NFA $\mathcal{A} = (\Sigma, S, S_0, E,S_F)$,
 $N$ be a positive integer, and
 $\mathcal{M}_{\mathcal{A},N}=(\Sigma_{\bot},\Sigma_{\bot},Q,q_0,\Delta,\Lambda)$
 be the filter Moore machine in Def.~\ref{def:filterConstructionUntimed}.
Let  $w=a_1 a_2\dotsc a_n$ be a word over $\Sigma$, and 
 $\bot^{N} w'$ be the
 output word of $\mathcal{M}_{\mathcal{A},N}$ over the input word
 $w \bot^{N}$. Let $w'=b_1 b_2\dots b_n$, where $b_{i}\in\Sigma_{\bot}$. 

\begin{theorem}
 [soundness]
 \label{thm:nfa_soundness}
The Moore machine  $\mathcal{M}_{\mathcal{A},N}$ is sound, in the sense of Def.~\ref{def:filterUntimed}. 
 \qed
\end{theorem}

Completeness holds if, the length of matches is bounded
and the buffer size $N$ is no shorter than the bound. This is essentially
 the setting of \emph{multiple string matching}~\cite{DBLP:journals/cacm/AhoC75}.
\begin{theorem}
 [completeness]
 \label{thm:nfa_completeness}
 Assume we have
 $\max\{|w| \mid w \in L(\mathcal{A})\} \leq N < \infty$. Then 
 we can construct an NFA $\mathcal{A}'$ with
 $L(\mathcal{A}) = L(\mathcal{A}')$, so that the filter Moore machine
 $\mathcal{M}_{\mathcal{A}',N}$  is complete. The latter means: 
 if the index $k$ satisfies  $a_k = b_k$, then there is an interval $[i,j]$ such that
 $k\in [i,j]$ and $w|_{[i,j]}\in L(\mathcal{A})$.
 \qed
\end{theorem}

\auxproof{
Thm.~\ref{thm:nfa_substring_window} shows that at worse, our procedure filters out
as many characters as the following sliding window-based procedure: it
tries the matching containing each substring of the same length as the
shortest match, and it filters out a character $a_k$ if all of the
matching trial containing $a_k$ fails.

\begin{theorem}
 \label{thm:nfa_substring_window}
 For any index $k\in [1,|w|]$ of $w$, if we have $\bigcup_{d\in [0,M]}
 (\Sigma^* w|_{[k-d,k-d+M-1]} \Sigma^*) \cap
 L(\mathcal{A})=\emptyset$, $b_k=\bot$ holds.
 \qed
\end{theorem}
}

The intuition for monotonicity is that, the bigger the buffer size $N'$ is, the filter
 $\mathcal{M}_{\mathcal{A},N'}$ masks more characters while its state space grows. 
A precise statement is more intricate, requiring the bigger buffer size $N''$ to be a multiple of the smaller one. 
\begin{theorem}
 [monotonicity]
 \label{thm:nfa_monotonicity}
 For any positive integer $N'$, let
 $\mathcal{M}_{\mathcal{A},N'}
$
 be the filter Moore machine of Def.~\ref{def:filterConstructionUntimed},
 and $\bot^{N'} w'^{(N')}$ be the
 output word of $\mathcal{M}_{\mathcal{A},N'}$ over the input word
 $w \bot^{N'}$. Let  $w'^{(N')}=b^{(N')}_1 b^{(N')}_2 \dotsc b^{(N')}_n$ where $b^{(N')}_{i}\in\Sigma_{\bot}$. 
 For any positive integers $n,N'$ and any index $k$ of $w$, $b^{(N')}_{k}=\bot$ implies
 $b^{(nN')}_{k}=\bot$.
 \qed
\end{theorem}

As mentioned in Prop.~\ref{prop:filter_state_space}, the state space of
our filter Moore machine is exponentially bigger than that of  $\mathcal{A}$. This is because of the
powerset construction required by its deterministic branching.
By sacrificing the execution time, one can determinize the NFA on the
fly, which usually needs less memory space.
See Appendix~\ref{appendix:nfa_filtering}
of~\cite{AnonymousExtendedVersion}.

\section{Moore-Machine Filtering for Pattern Matching II: Timed}
\label{sec:timed}
We present our main contribution, that is, the construction of filter Moore machines for \emph{timed} pattern matching. While the basic ideas stay the same as in the untimed setting (\S{}\ref{sec:untimed}), \emph{determinization} poses a technical challenge, since timed automata (TA) cannot determinized in general~\cite{DBLP:journals/ipl/Tripakis06}. Here we rely on  \emph{one-clock determinization}~\cite[\S{}5.3]{DBLP:journals/fmsd/KrichenT09}. The construction overapproximates reachability, hence we  maintain soundness of filtering. Moreover, the local nature of the resulting TA---it has only one clock variable that is reset at every transition---allows us to devise a filter that is an \emph{untimed}, \emph{finite-state} Moore machine. 

\subsection{Problem Formulation}

\begin{definition}[timed words]
 Let $\Sigma$ be an alphabet. A \emph{timed word} over $\Sigma$ is a
 sequence $w$ of pairs $(a_i,\tau_i) \in \Sigma \times \Rpos$
 satisfying $\tau_i < \tau_{i + 1}$ for any $i \in [1,|w|-1]$.
 Let $w = (\overline{a},\overline{\tau})$ be a timed word.
The set of timed words over $\Sigma$ is
 denoted by $\mathcal{T}(\Sigma)$.

 We denote the subsequence $(a_i, \tau_i),(a_{i+1},
 \tau_{i+1}),\cdots,(a_j,\tau_j)$ by $w (i,j)$.
 For $t \in \Rnn$, the \emph{$t$-shift} of $w$ is
 $(\overline{a}, \overline{\tau}) + t = (\overline{a}, \overline{\tau} +
 t)$ where
 $\overline{\tau} + t = (\tau_1 + t,\tau_2 + t,\cdots, \tau_{|\tau|} + t)$.
 For timed words $w = (\overline{a},\overline{\tau})$ and
 $w' = (\overline{a'},\overline{\tau'})$,
 their \emph{absorbing concatenation} is
 $w \circ w' = (\overline{a} \circ \overline{a'}, \overline{\tau} \circ
 \overline{\tau'})$ where $\overline{a} \circ \overline{a'}$ and 
 $\overline{\tau} \circ \overline{\tau'}$ are usual concatenations, and 
 their \emph{non-absorbing concatenation} is
 $w \cdot w' = w \circ (w' + \tau_{|w|})$.
 We note that the absorbing concatenation $w \circ w'$ is defined only
 when $\tau_{|w|} < \tau'_{1}$.

 For a timed word $w = (\overline{a}, \overline{\tau})$ on $\Sigma$
 and 
 $t,t' \in \R_{>0}$ satisfying $t < t'$, a \emph{timed word segment}
 $w|_{(t,t')}$ is defined by the timed word $(w(i,j) - t) \circ (\$,t'-t)$
 on the augmented alphabet $\Sigma \amalg \{\$\}$,
 where $i,j$ are chosen so that
 $\tau_{i-1} \leq t < \tau_i$ and
 $\tau_{j} < t' \leq \tau_{j+1}$.
 Here the timed word  $w(i,j) - t$ is the $(-t)$-shift of $w(i,j)$; and  the fresh symbol ${\$}$ is called the \emph{terminal character}. 

\end{definition}

\begin{definition}[timed automaton]
\label{def:semantics_ta}
Let $C$ be  a finite set of \emph{clock variables}, and $\Phi (C)$ denote the set of
 conjunctions of inequalities $x \bowtie c$ where $x \in C$, $c \in
 \Z_{\geq 0}$, and ${\bowtie} \in \{>,\geq,<,\leq\}$.
 A \emph{timed automaton}
 $\mathcal{A} = (\Sigma,S,s_0,S_{F},C,E)$
is a tuple where $\Sigma$ is an alphabet,
 $S$ is a finite set of states, $s_0 \in S$ is an initial state,
 $S_{F} \subseteq S$ is a set of accepting
 states, and
 $E \subseteq S \times S \times \Sigma \times \mathcal{P}(C) \times \Phi(C)$
 is a set of transitions. The components of a transition $(s,s',a,\lambda,\delta)\in E$ represent:  the source,  target,  action, reset variables and guard of the transition, respectively.

We define a \emph{clock valuation} 
 $\nu$ as a function $\nu: C \to \R_{\geq 0}$.
We define the \emph{$t$-shift} $\nu + t$  of a clock
 valuation $\nu$, where $t \in \R_{\geq 0}$, by  $(\nu + t) (x) = \nu (x) + t$ for any $x \in C$.
 For a timed automaton $\mathcal{A} = (\Sigma,S,s_0,S_{F},C,E)$ and a timed
 word $w = (\overline{a},\overline{\tau})$, a \emph{run} of $\mathcal{A}$
 over $w$ is a sequence $r$ of pairs 
 $(s_i, \nu_i) \in S \times (\R_{\geq 0})^C$ satisfying the following:
 (initiation) $s_0$ is the initial state  and
 $\nu_0 (x) = 0$ for any $x \in C$; and
 (consecution) for any $i \in [1,|w|]$, there
 exists a transition $(s_{i-1}, s_i, a_i, \lambda, \delta) \in E$
 such that $\nu_{i-1} + \tau_i - \tau_{i-1} \models \delta$ and
 $\nu_i (x) = 0$  (for $x \in \lambda$) and 
 $\nu_i (x) = \nu_{i-1} (x) + \tau_i - \tau_{i-1}$ (for $x \not\in \lambda$).
 A run only satisfying the consecution condition is a \emph{path}.
 A run $r = (\overline{s},\overline{\nu})$ is \emph{accepting} if the last element $s_{|s|-1}$ of $s$ belongs to $S_{F}$. 
 The \emph{language} $L (\mathcal{A})$
 is defined to be the set
 $\{w \mid \text{ there is an accepting run of $\mathcal{A}$ over $w$}\}$ of timed words.
\end{definition}

Here is our target problem. Its algorithms  have been actively studied~\cite{DBLP:conf/formats/UlusFAM14,DBLP:conf/tacas/UlusFAM16,DBLP:conf/formats/WagaHS17}; filtering Moore machines as  preprocessors for those algorithms is this paper's contribution. 

\begin{definition}[timed pattern matching]\label{def:TimedPatternMatching}
Let $\mathcal{A}$
be a timed automaton, and  $\str$ be a timed word,  over a common  alphabet $\Sigma$. The \emph{timed pattern matching} problem requires
all the intervals $(t,t')$ for which the segment  $w|_{(t,t')}$ is 
 accepted by $\mathcal{A}$. That is, it requires
 the \emph{match set} $\Match
 (w,\mathcal{A}) = \{(t,t') \mid w|_{(t,t')} \in L (\mathcal{A})\}$.
\end{definition}

\subsection{One-Clock Determinization of TA}
Among the three main building blocks for our untimed construction of filter (Def.~\ref{def:filterConstructionUntimed}), \emph{counters} and a \emph{buffer} carry over smoothly to the current timed setting. For \emph{determinization} we rely on the overapproximating notion in Def.~\ref{def:oneClockDet}; it is taken from~\cite[\S{}5.3]{DBLP:journals/fmsd/KrichenT09}.

We start with some auxiliary notations. 
\begin{definition}[restriction $\nu|_{C}$, join $\nu\sqcup\nu'$]
 Let $\nu\colon C' \to \R_{\geq 0}$ be a clock valuation. The  restriction of $\nu$
to  $C\subseteq C'$ is denoted by $\nu|_{C}\colon C\to \R_{\geq 0}$. That is, $(\nu|_{C})(x)=\nu(x)$ for each $x\in C$. 

 Let $\nu\colon C \to \R_{\geq 0}$ and  $\nu'\colon C' \to \R_{\geq 0}$ be clock variations. We define their \emph{join} $\nu\sqcup\nu'\colon C\amalg C'\to \R_{\geq 0}$ to be the following clock valuation over the disjoint union $C\amalg C'$. 
\begin{align*}
(\nu\sqcup\nu')(x)=
\begin{cases}
 \nu(x) &\text{if $x\in C$}
\\
 \nu'(x) &\text{if $x\in C'$.}
\end{cases}
\end{align*}

 The function that maps $x_{i}$ to $r_{i}$ (for each $i\in \{1,\dotsc,n\}$) is denoted by $[x_{1}\mapsto r_{1},\dotsc, x_{n}\mapsto r_{n}]$.
\end{definition}

\begin{definition}[one-clock determinization] \label{def:oneClockDet}
 Let $\mathcal{A}=(\Sigma,S,s_0,S_F,C,E)$ be a timed automaton (TA), and $y$ be a fresh clock variable (i.e.\ $y\not\in C$). 
A TA
 $\mathcal{A}'=(\Sigma,S',s'_0,S'_F,\{y\},E')$
is a \emph{one-clock determinization} of $\mathcal{A}$ if the following conditions hold. 
\begin{enumerate}
 \item Each element $\mathcal{S}\in S'$ of the (new, finite) state space $S'$ is a finite set $\mathcal{S}=\{(s_{1},Z_{1}),\dotsc, (s_{m},Z_{m})\}$ of pairs $(s_{i},Z_{i})$, where $s_{i}\in S$ is a state of $\mathcal{A}$, and $Z_{i}$ is a subset of $(\R_{\geq 0})^{C\amalg\{y\}}$ given by a special polytope called a \emph{zone} (see e.g.~\cite{BehrmannBLP06}). 
 \item For each transition $(\mathcal{S},a,\delta,\lambda,\mathcal{S}')\in E'$ of $\mathcal{A}'$, the guard $\delta$ is  a finite union of intervals of $y$. Moreover it respects enabledness of transitions $E$ in $\mathcal{A}$. Precisely, for any $u,u'\in \R_{\geq 0}$ that satisfy $\delta$, we have $E_{a}(\mathcal{S},u)=E_{a}(\mathcal{S},u')$, where the set $E_{a}(\mathcal{S},u)\subseteq E$ is defined by
\begin{multline*}
 E_{a}(\mathcal{S},u)
 =\bigl\{\,(s,a,\delta',\lambda',s') \in E\;\big|\;
\\ \exists (s,Z)\in \mathcal{S}.\, \exists \nu\in Z.\, \nu(y)=u \text{ and $\nu$ satisfies $\delta'$} \,\bigr\}\enspace.
\end{multline*}

 \item\label{item:yIsReset} Each transition of $\mathcal{A}'$ resets the unique clock variable $y$. That is, for each transition $(\mathcal{S},a,\delta,\lambda,\mathcal{S}')\in E'$, 
we have
       $\lambda = \{y\}$.
 \item Each transition $(\mathcal{S},a,\delta,\lambda,\mathcal{S}')\in E'$ of $\mathcal{A}'$ simulates transitions of $\mathcal{A}$. More precisely, let $(s,Z)\in \mathcal{S}$ and $(\nu\colon C\amalg\{y\}\to \R_{\geq 0}) \in Z$. Assume that $(s,\nu|_{C})\xrightarrow{a,\tau}(s',\nu')$ is a (length-1) path of $\mathcal{A}$, for some $s'\in S$ and $\nu'\colon C\to \Rnn$ (here $\tau$ is the dwell time). Then we require that there exist a zone $Z'\subseteq (\R_{\geq 0})^{C\amalg\{y\}}$ such that 1) $(s',Z')\in \mathcal{S}'$ and 2) the valuation $\nu'\sqcup [y\mapsto \tau]$, over the clock set $C\amalg\{y\}$, belongs to $Z'$.

 \item $\mathcal{A}'$ is deterministic: for each state $\mathcal{S}\in S'$, 
each clock valuation $\nu\in (\R_{\geq 0})^{\{y\}}$,
$a\in \Sigma$ and $\tau\in \Rnn$ for the dwell time, a length-1 path
from $(\mathcal{S},\nu)$ labeled with $a,\tau$ is unique. That is, if
$(\mathcal{S},\nu)\xrightarrow{a,\tau} (\mathcal{S}',\nu')$ and
$(\mathcal{S},\nu)\xrightarrow{a,\tau} (\mathcal{S}'',\nu'')$ are both paths in $\mathcal{A}'$, then $\mathcal{S}'=\mathcal{S}''$ and $\nu'=\nu''$. (Note that Cond.~\ref{item:yIsReset} forces $\nu'=\nu''=[y\mapsto 0]$.)
 \item The initial state $s'_{0}$ of $\mathcal{A}'$ is given by $s'_{0}=\{(s_{0}, \{\mathbf{0}\})\}$. Here $\mathbf{0}$ is the valuation that maps every clock variable to $0$. 
 \item A state $\mathcal{S}$ belongs to $S'_{F}$ if and only if there exists $(s,Z)\in \mathcal{S}$ such that $s\in S_{F}$. 
\end{enumerate}
\end{definition}


\begin{proposition}\label{prop:propertiesOfOneClockDeterminization}
Let $\mathcal{A}=(\Sigma,S,s_0,S_F,C,E)$ be a TA, and
 $\mathcal{A}'=(\Sigma,S',s'_0,S'_F,\{y\},E')$
be a one-clock determinization of $\mathcal{A}$. 
 Then $\mathcal{A}'$
satisfies the following properties. 
\begin{itemize}
 \item  (Simulation) Let $w\in\mathcal{T}(\Sigma)$ be a timed word, and assume that there exists a run over $w$ to a state $s\in S$ in $\mathcal{A}$. Then there exists $\mathcal{S}\in S'$ that satisfies the following: 1) $(s,Z) \in \mathcal{S}$ for some zone $Z$, and 2) there exists a run over $w$ to $\mathcal{S}$ in $\mathcal{A}'$. 
 \item (Language inclusion) In particular, $L(\mathcal{A})\subseteq L(\mathcal{A}')$.
\qed
\end{itemize}
\end{proposition}

Note that Def.~\ref{def:oneClockDet} gives a property and not a
construction: there are multiple one-clock determinizations---of varying
size and precision---of the same TA $\mathcal{A}$. In our implementation
we use a specific construction presented
in~\cite[\S{}5.3.4]{DBLP:journals/fmsd/KrichenT09}. 
A sketch is in Appendix~\ref{appendix:one_clock_determinization}
of~\cite{AnonymousExtendedVersion}.

\subsection{Construction of Our Filter $\mathcal{M}_{\mathcal{A},N}$}
\begin{definition}[the filter $\mathcal{M}_{\mathcal{A},N}$ for timed pattern matching]
\label{def:filterConstructionTimed}
Let  $\mathcal{A} = (\Sigma, S, s_0, S_F, C, E)$ be a TA, and $N\in\Zp$.
The construction of the filter Moore machine $\mathcal{M}_{\mathcal{A},N}$ is by the following steps.

In the first step we augment the original TA
 $\mathcal{A}$ with counters. Specifically, let $\mathcal{A}^{N\text{-}\mathrm{ctr}}
=(\Sigma_{\bot}, S\times
[0,N]
,(s_0,0), S^{N\text{-}\mathrm{ctr}}_F, C,E^{N\text{-}\mathrm{ctr}})$ be defined as follows.
 $S^{N\text{-}\mathrm{ctr}}_F = \{(s_f,n)\mid s_f\in S_F, n \in
[0,N]
\}$, and
 \begin{align*}\small
\begin{array}{l}
   E^{N\text{-}\mathrm{ctr}} = \bigl\{\,\bigl((s_0,0),a,\mathbf{true},C,(s_0,0)\bigr) \;\big|\; a \in \Sigma_{\bot}\,\bigr\}\\
  \quad\cup\bigl\{\,\bigl((s,n),a,\delta,\lambda,(s',n+1)\bigr)\;\big|\;
  (s,a,\delta,\lambda,s')\in E, n \in [0,N-1]\,\bigr\}\\
  \quad\cup \bigl\{\,\bigl((s,N),a,\delta,\lambda,(s',1)\bigr)\;\big|\;
  (s,a,\delta,\lambda,s')\in E\,\bigr\}\enspace.
\end{array}
 \end{align*}

In the second step we take a one-clock determinization (Def.~\ref{def:oneClockDet}) of $\mathcal{A}^{N\text{-}\mathrm{ctr}}$.
Let 
\begin{math}
 \mathcal{A}^{N\text{-}\mathrm{ctr\text{-}d}}=
(\Sigma_{\bot}, S^{N\text{-}\mathrm{ctr\text{-}d}}, s_{0}^{N\text{-}\mathrm{ctr\text{-}d}},S_{F}^{N\text{-}\mathrm{ctr\text{-}d}}, \{y\}, E^{N\text{-}\mathrm{ctr\text{-}d}})\end{math}
be the outcome. 

Finally in the third step we define the Moore machine $\mathcal{M}_{\mathcal{A},N}$:
\begin{align*}
\begin{array}{r}
 \mathcal{M}_{\mathcal{A},N} 
 \;=\;
 \Bigl(\,\Sigma_{\bot}\times \Rnn,\, \{\pass,\mask\},\, S^{N\text{-}\mathrm{ctr\text{-}d}}\times \{\pass,\mask\}^{N},\,
 \\
 \bigl(s_{0}^{N\text{-}\mathrm{ctr\text{-}d}},(\mask,\dotsc,\mask)\bigr), \,
 \Delta,\, \Lambda
 \,\Bigr)\enspace,
\end{array}
\end{align*}
where $\Delta$ and $\Lambda$ are defined as follows. 
\begin{math}
 \Delta\bigl(\,(\mathcal{S}, \overline{l}),\, (a,\tau)\,\bigr)
 = (\mathcal{S'}, \overline{l}')
\end{math}, where $\mathcal{S'}$ is the unique successor of $\mathcal{S}$ in 
$\mathcal{A}^{N\text{-}\mathrm{ctr\text{-}d}}$
under the character $a$ and the dwell time $\tau$ (Def.~\ref{def:oneClockDet}), and $\overline{l}'$ is defined as follows.
\begin{displaymath}\footnotesize
\begin{array}{l}
  \overline{l}' = 
 \begin{cases}
  \pass^{N}&\text{if  $\exists s, Z.\,  \bigl((s, N),Z\bigr) \in \mathcal{S}$}\\
  l_2,l_3,\dots,l_{N-\psi(\mathcal{S}')+1},\overbrace{\pass,\dotsc,\pass}^{\psi(\mathcal{S}')}&
  \text{else if 
    $\exists ((s,n),Z) \in \mathcal{S}.\, s \in S_F$
  }\\
  l_2,l_3,\dots,l_{N},\mask&\text{otherwise}
 \end{cases}
\end{array}
\end{displaymath}
Here $\psi(\mathcal{S}')=\max\{n\mid \exists s, Z.\,  ((s,n),Z) \in \mathcal{S} \text{ and } s \in S_F \}$. We define \begin{math}
 \Lambda\bigl((\mathcal{S},(l_1,l_2,\dots,l_{N}))\bigr)=l_1
\end{math}.
\end{definition}

\begin{wrapfigure}{r}{0pt}
 \includegraphics[clip,width=.50\linewidth]{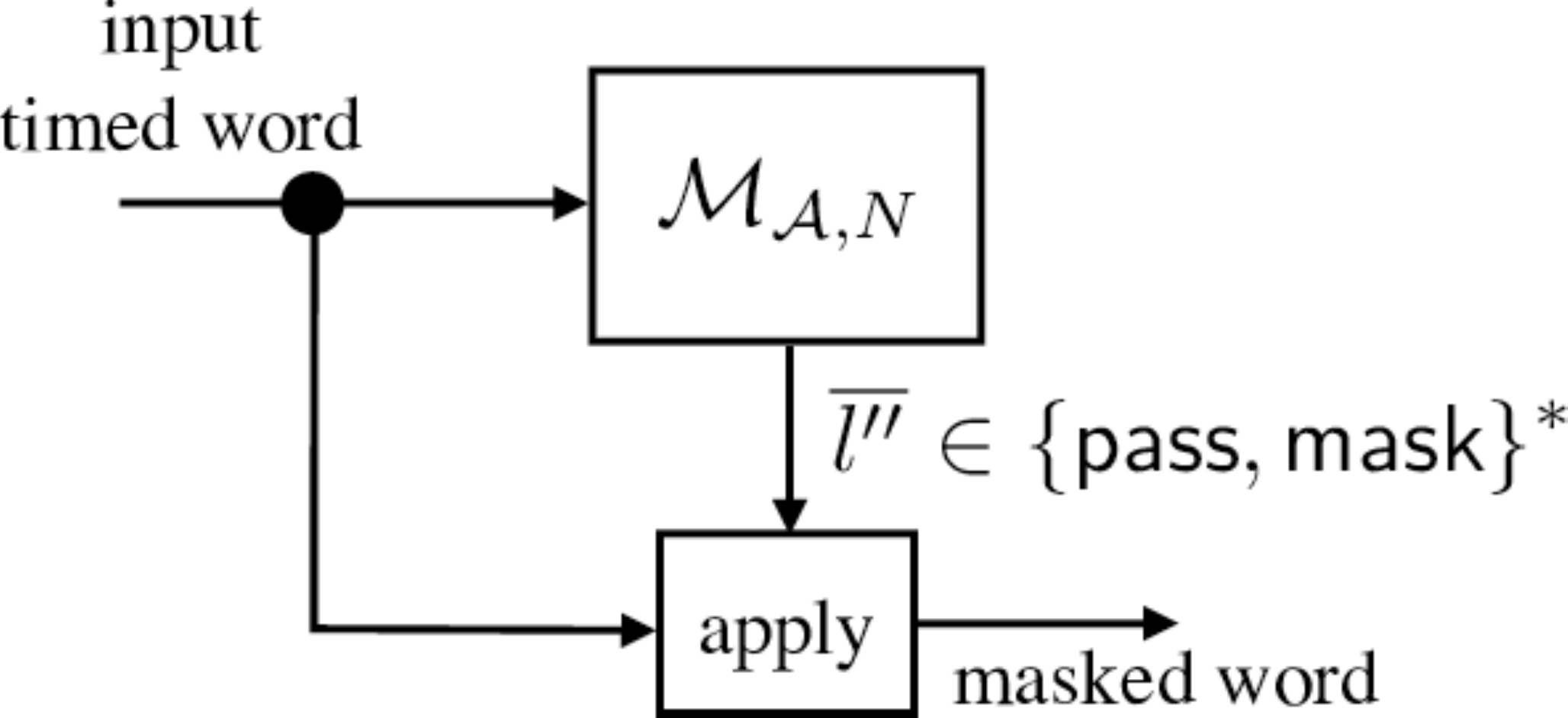}
\end{wrapfigure}
Note that the resulting Moore machine takes timed words as input.
This makes its input alphabet infinite (namely $\Sigma_{\bot}\times \Rnn$). This is not a big issue in implementation: we note that the state space is still finite. Moreover, the output alphabet of $\mathcal{M}_{\mathcal{A},N}$ is a two-element set $\{\pass,\mask\}$, meaning that the filter Moore machine only outputs the masking information. A finite-state Moore machine is certainly not capable of buffering time-stamped characters, but the original timed word can be copied and suitable masking can be applied to it later. See the above figure. 

\begin{theorem}[soundness]
 \label{thm:timed_soundness}
 Let $\mathcal{A}$ be a pattern TA $\mathcal{A} = (\Sigma, S, s_0, S_F, C, E)$,
 $N$ be a positive integer, and
\begin{math}
 \mathcal{M}_{\mathcal{A},N}
\end{math}
 be the filter Moore machine in Def.~\ref{def:filterConstructionTimed}.
 Let $w=(a_1,\tau_1) (a_2,\tau_2)\dotsc (a_n,\tau_n)$ be a timed word
 over $\Sigma$, and 
 $\mask^{N} w'$ be the
 output word of $\mathcal{M}_{\mathcal{A},N}$ over the input word
 $w (\bot,\tau_{n})^{N}$
(here the input and output words are padded by $\mask^{N}$ and $(\bot,\tau_{n})^{N}$, much like in Fig.~\ref{fig:filteringPadding}). Let $w'=b_1 b_2\dots b_n$, where
 $b_{k}\in\{\pass,\mask\}$.

 For any pair $(i,j)$ of indices of $w$ satisfying
 $w(i,j) - \tau_{i-1} \in L(\mathcal{A})$ and for any index
 $k\in[i,j]$ of $w$, we have $b_k = \pass$.
 \qed
\end{theorem}

\section{Implementation and Experiments}\label{sec:experiments}

We implemented our filter construction for \emph{timed} pattern matching.
Our implementation suppresses successive $\bot$'s into two $\bot$'s,
maintaining the timestamp of the first and the last $\bot$'s.
We generate the buffer part of the state space $Q$ on-the-fly
(namely $\{\pass,\mask\}^N$ of
Def.~\ref{def:filterConstructionTimed}); this is as we discussed before Prop.~\ref{prop:filter_state_space}.
We conducted experiments to answer the following research questions.
\begin{itemize}
 \item[RQ1:] Does our filter Moore machine mask many events?
 \item[RQ2:] 
Is our filter Moore machine online capable?
       That is, does it work in linear time
       and  constant space, with respect to the length of the input timed word?
 \item[RQ3:] Does our filter Moore machine accelerate the whole task of timed pattern matching?

 \item[RQ4:] Is our filter precise? That is, do many of the non-masked events  contribute to actual matches?
 \item[RQ5:] Is our filter responsive? That is, does it not cause large delays?
\end{itemize}
We implemented our filter construction in C++ and compiled them by clang-900.0.39.2. The tool's input consists of a pattern TA $\mathcal{A}$, the buffer size $N$ and a timed word $w$; it outputs a filtered word.    The experiments
are done on Mac OS 10.13.4, MacBook Pro Early 2013 with 2.6 GHz Intel
Core i5 processor and 8 GB 1600MHz DDR3 RAM. The benchmark problems we
used are in Fig.~\ref{fig:case1_pattern}---\ref{fig:case3_pattern}.
All of them are from automotive scenarios.

For the measurement of the execution time and the memory usage,  we used
GNU time and took the average of 20 executions.
In each experiment, we measured the whole workflow: 
in RQ2, the time and memory usage include that for constructing a filter; and in 
 RQ3, the time and memory usage account for filter construction, filtering, process intercommunication, and pattern matching.
In the experiments for RQ3, we used \monaa~\cite{MONAACode},  a recent tool for timed pattern matching.
See Appendix~\ref{appendix:detailed_results}
of~\cite{AnonymousExtendedVersion} for the detailed results.

\begin{figure}[tbp]
\begin{minipage}{\linewidth}
  \centering
 \scalebox{0.57}{
 \begin{tikzpicture}[shorten >=1pt,node distance=2.35cm,on grid,auto] 
    \node[state,initial] (s_0)  {$s_0$}; 
    \node[state,node distance=2cm] (s_1) [right=of s_0] {$s_1$}; 
    \node[state] (s_2) [right=of s_1] {$s_2$};
    \node[state] (s_3) [right=of s_2] {$s_3$};
    \node[state] (s_4) [right=of s_3] {$s_4$}; 
    \node[state] (s_5) [right=of s_4] {$s_5$}; 
    \node[state,node distance=2cm,accepting] (s_6) [right=of s_5] {$s_6$};
    \path[->] 
    (s_0) edge  [above] node {\begin{tabular}{c}
                               $\textrm{low},\mathbf{true}$\\
                               $/x := 0$
                              \end{tabular}} (s_1)
    (s_1) edge  [above] node {\begin{tabular}{c}
                               $\textrm{high},$\\
                               $0 < x < 1$
                              \end{tabular}} (s_2)
    (s_2) edge  [above] node {\begin{tabular}{c}
                               $\textrm{high},$\\
                               $0 < x < 1$
                              \end{tabular}} (s_3)
    (s_3) edge  [above] node {\begin{tabular}{c}
                               $\textrm{high},$\\
                               $0 < x < 1$
                              \end{tabular}} (s_4)
    (s_4) edge  [above] node {\begin{tabular}{c}
                               $\textrm{high},$\\
                               $0 < x < 1$
                              \end{tabular}} (s_5)
    (s_5) edge  [above] node {\begin{tabular}{c}
                               $\textrm{high},$\\
                               $1 < x$
                              \end{tabular}} (s_6)
    (s_5) edge  [above, loop above] node {$\textrm{high},\mathbf{true}$} (s_5);
 \end{tikzpicture}}

 \caption{\textsc{Torque}. The set $W$ (length 242,808--4,873,207) of input words is
 generated by the automotive engine model \texttt{sldemo\_enginewc.slx}~\cite{SimulinkGuide} 
with random input.
 The pattern specifies four or more consecutive occurrences of
 $\textrm{high}$ in one second.
 For $N=10$, the size of our filter Moore machine $\mathcal{M}_{\mathcal{A},10}$, measured by the number 
of the reachable states in the non-buffer part $S^{N\text{-}\mathrm{ctr\text{-}d}}$ (cf.\
 Def.~\ref{def:filterConstructionTimed} and Prop.~\ref{prop:filter_state_space}), is $16$.
 }
 \label{fig:case1_pattern}
\end{minipage}
\end{figure}
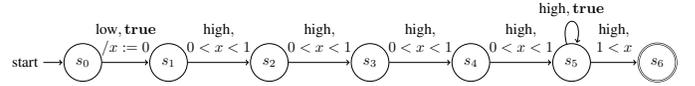

\begin{figure}[tbp]
\begin{minipage}{\linewidth}
   \centering
  \scalebox{0.6}{
  \begin{tikzpicture}[shorten >=1pt,node distance=2.5cm,on grid,auto] 
    \node[state,initial] (s_0)  {}; 
    \node[state,node distance=2.5cm] (s_1) [right=of s_0] {$\text{g}_1$}; 
    \node[state,node distance=2.5cm,accepting] (s_2) [right=of s_1] {$\text{g}_2$};
    \path[->] 
    (s_0) edge  [above] node {$\text{g}_1 /x := 0$} (s_1)
    (s_1) edge  [above] node {$\text{g}_2, x < 2$} (s_2);
  \end{tikzpicture}}
  \caption{\textsc{Gear}. The set $W$ (length 306--1,011,426) of input words is generated by the automatic transmission system model in~\cite{DBLP:conf/cpsweek/HoxhaAF14}. The pattern, from $\phi^{\mathit{AT}}_5$ in~\cite{DBLP:conf/cpsweek/HoxhaAF14}, is for an event in which gear shift occurs too quickly (from the 1st to 2nd).
 For $N=10$, the size of our filter Moore machine $\mathcal{M}_{\mathcal{A},10}$, measured in the same way of Fig.~\ref{fig:case1_pattern}, is $3$.
 }
  \label{fig:case2_pattern}
\end{minipage}
\end{figure}
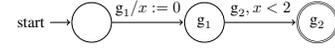

\begin{figure}[tbp]
\begin{minipage}{\linewidth}
   \centering
  \scalebox{0.45}{
  \begin{tikzpicture}[on grid,auto,node distance=2.5cm]
   \node[state, initial] (s_000) {$?$};

   \node[state] (s_100)[above right=of s_000] {$\text{g}_1$};
   \node[state] (s_001)[below right=of s_000] {$?$};

   \node[state] (s_200)[right=of s_100] {$\text{g}_2$};
   \node[state] (s_101)[right=of s_001] {$\text{g}_1$};

   \node[state] (s_300)[right=of s_200] {$\text{g}_3$};
   \node[state] (s_201)[right=of s_101] {$\text{g}_2$};

   \node[state] (s_400)[right=of s_300] {$\text{g}_4$};
   \node[state] (s_301)[right=of s_201] {$\text{g}_3$};

   \node[state] (s_401)[right=of s_301] {$\text{g}_4$};
   \node[state,accepting] (f)[right=3cm of s_401] {\cmark};
   \node[state] (f0)[above=1.75cm of f] {};

   \path[->]
   (s_000) edge  [above left] node {$\text{g}_1, \mathbf{true}$} (s_100)
   (s_100) edge  [above] node {$\text{g}_2, \mathbf{true}$} (s_200)
   (s_200) edge  [above] node {$\text{g}_3, \mathbf{true}$} (s_300)
   (s_300) edge  [above] node {$\text{g}_4, \mathbf{true}$} (s_400)

   (s_100) edge  [left] node {$\omega \geq 2500, \mathbf{true}$} (s_101)
   (s_200) edge  [left] node {$\omega \geq 2500, \mathbf{true}$} (s_201)
   (s_300) edge  [left] node {$\omega \geq 2500, \mathbf{true}$} (s_301)
   (s_400) edge  [left] node {\begin{tabular}{c}
                                     $\omega \geq 2500, x < 10$\\
                                     $/x := 0$
                                    \end{tabular}} (s_401)

   (s_001) edge  [above] node {$\text{g}_1, \mathbf{true}$} (s_101)
   (s_101) edge  [above] node {$\text{g}_2, \mathbf{true}$} (s_201)
   (s_201) edge  [above] node {$\text{g}_3, \mathbf{true}$} (s_301)
   (s_301) edge  [above] node {\begin{tabular}{c}
                                $\text{g}_4, x < 10$\\
                                $/x := 0$
                               \end{tabular}} (s_401)




   (s_000) edge  [below left] node {$\omega \geq 2500, \mathbf{true}$} (s_001)
   %



   (s_401) edge  [below] node {\begin{tabular}{c}
                                $\omega < 2500, x > 1$\\
                                $\omega \geq 2500, x > 1$\\
                                $\text{g}_3, x > 1$\\
                                $\text{g}_4, x > 1$
                               \end{tabular}} (f)

   (s_401) edge  [above left,pos=0.75] node {\begin{tabular}{c}
                                     $\omega < 2500 /x := 0$\\
                                     $\omega \geq 2500 /x := 0$\\
                                     $\text{g}_3 /x := 0$\\
                                     $\text{g}_4 /x := 0$
                                    \end{tabular}} (f0)

   (f0) edge  [loop above] node {\begin{tabular}{c}
                                  $\omega < 2500, \mathbf{true}$\\
                                  $\omega \geq 2500, \mathbf{true}$\\
                                  $\text{g}_3, \mathbf{true}$\\
                                  $\text{g}_4, \mathbf{true}$
                                 \end{tabular}} (f0)


   (f0) edge  [right,pos=0.75] node {\begin{tabular}{c}
                             $\omega < 2500, x > 1$\\
                             $\omega \geq 2500, x > 1$\\
                             $\text{g}_3, x > 1$\\
                             $\text{g}_4, x > 1$
                            \end{tabular}} (f);

  \end{tikzpicture}}
  \caption{\textsc{Accel}.
 The set $W$ (length 708--1,739,535) of input words is generated by the same
 model as in \textsc{Gear}. The pattern
 is from
 $\phi^{\mathit{AT}}_8$ 
 in~\cite{DBLP:conf/cpsweek/HoxhaAF14}:  the gear shifts from 1st
 to 4th and RPM gets high 
  in its course, but the
 velocity is 
 low (i.e.\ the event $v\ge 100$ is absent). 
 For $N=10$, the size of our filter Moore machine $\mathcal{M}_{\mathcal{A},10}$, measured in the same way of Fig.~\ref{fig:case1_pattern}, is $71$.
 }
  \label{fig:case3_pattern}
\end{minipage}
\end{figure}
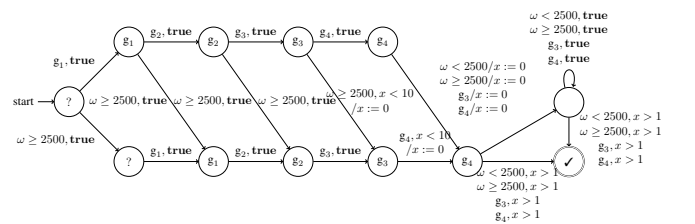

\begin{figure*}[tbp]
\noindent
\begin{minipage}{.33\linewidth}
  \centering
 \textsc{Torque}
 \scalebox{0.60}{\includegraphics{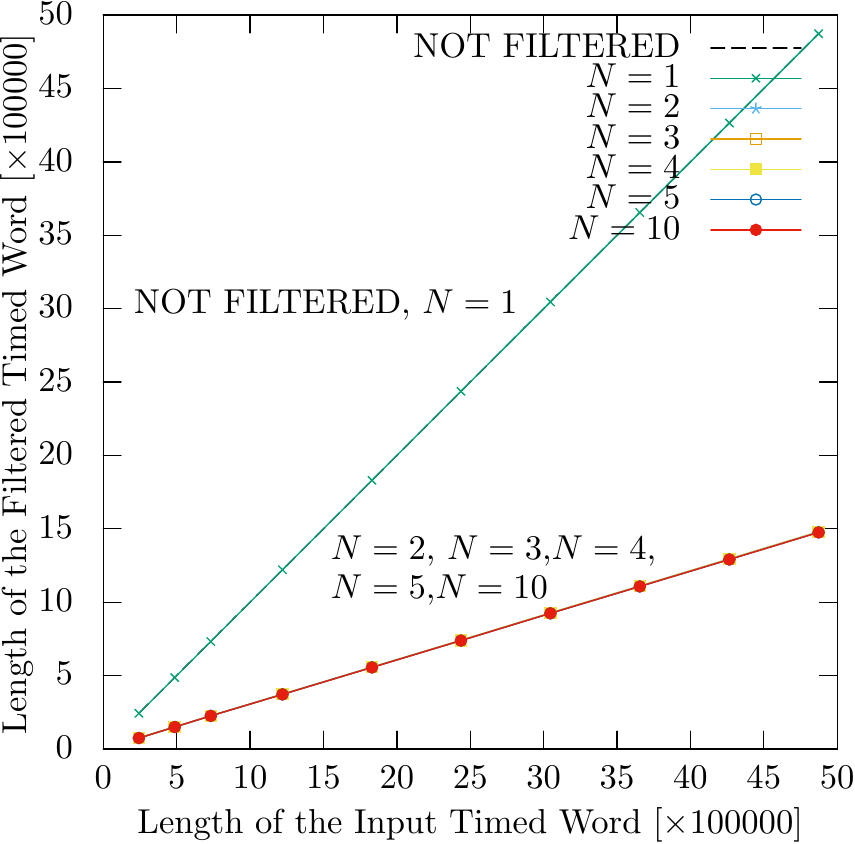}}
 \label{fig:filtering_rate_torque}
\end{minipage} 
\hspace{-2em}
\begin{minipage}{.33\linewidth}
 \centering
 \textsc{Gear}
 \scalebox{0.60}{\includegraphics{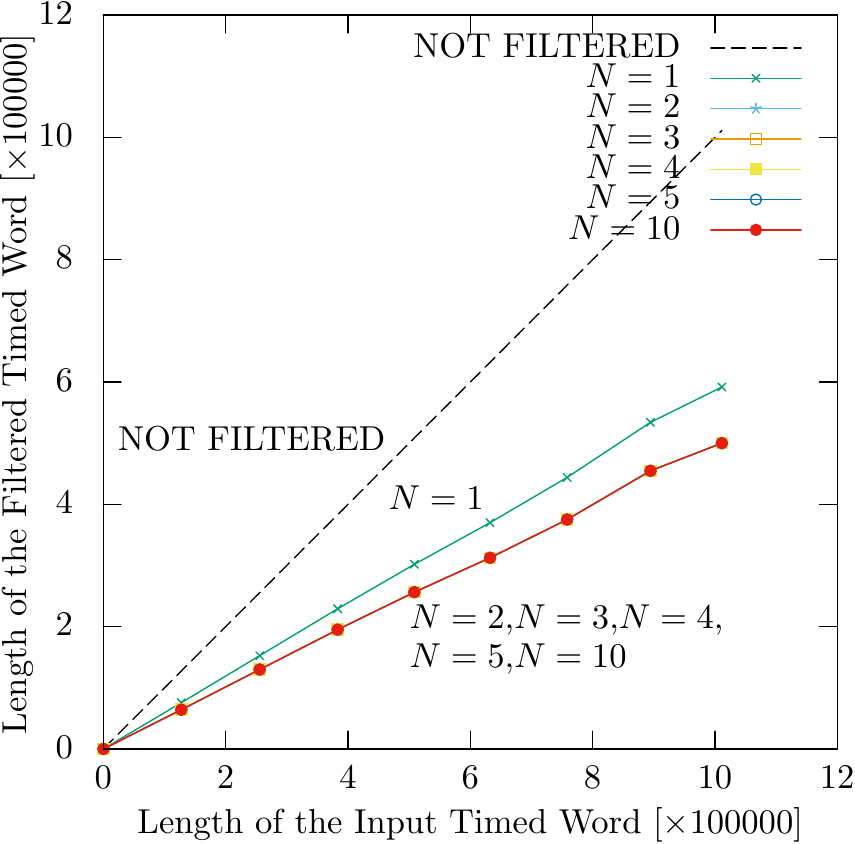}}
 \label{fig:filtering_rate_gear}
\end{minipage} 
\hspace{-2em}
\begin{minipage}{.33\linewidth}
 \centering
 \textsc{Accel}
 \scalebox{0.60}{\includegraphics{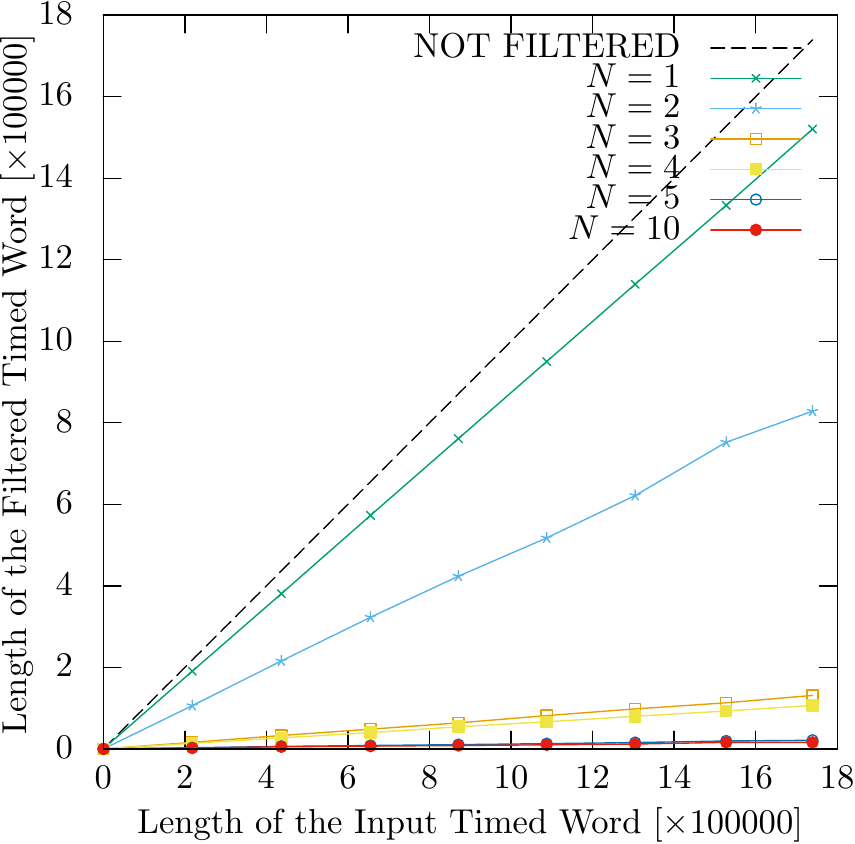}}
 \label{fig:filtering_rate_accel}
\end{minipage}
 \caption{Length of the input and filtered timed words}
 \label{fig:filtering_rate}
\end{figure*}
\begin{figure*}[tbp]
 \begin{minipage}{.33\linewidth}
  \centering
 \textsc{Torque}
\\
  \scalebox{0.45}{\includegraphics{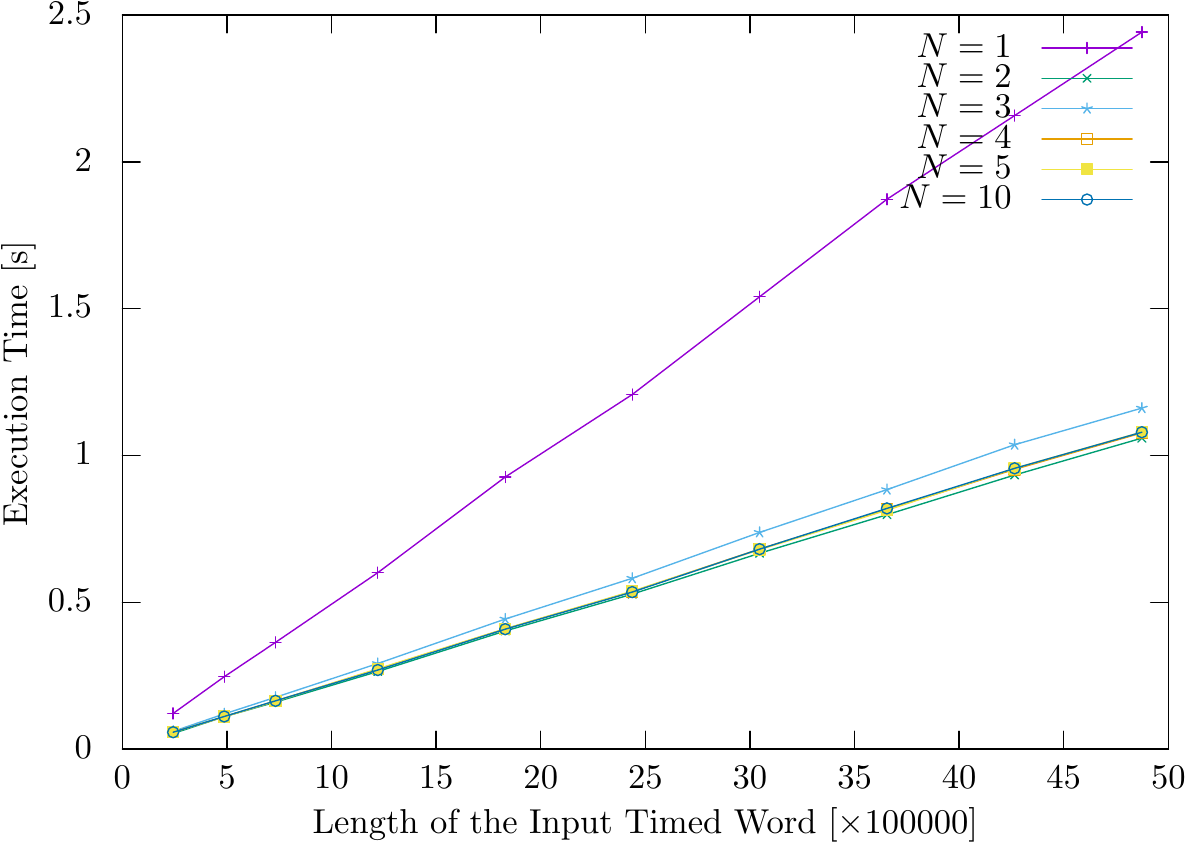}}
  \label{fig:execution_time_torque}
 \end{minipage}
 \begin{minipage}{.33\linewidth}
  \centering
 \textsc{Gear}
\\
  \scalebox{0.45}{\includegraphics{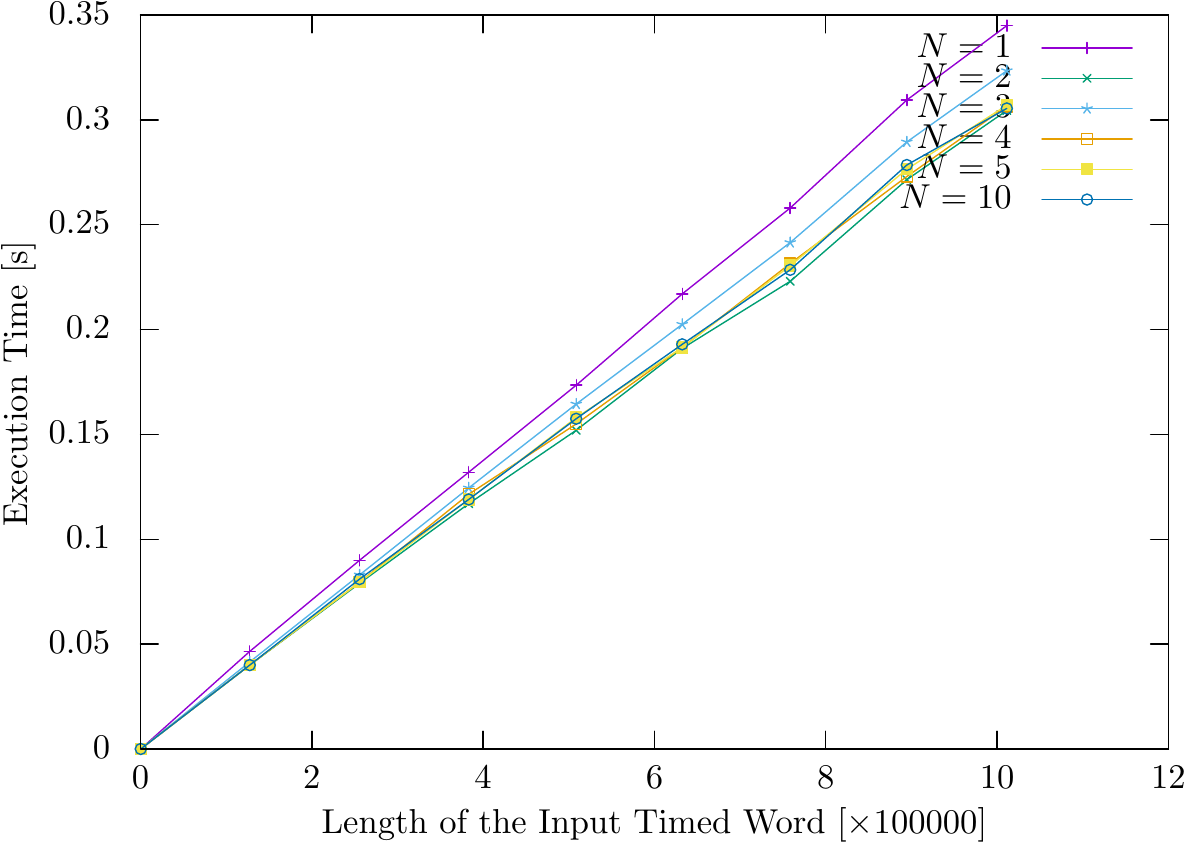}}
  \label{fig:execution_time_gear}
 \end{minipage}
 \begin{minipage}{.33\linewidth}
  \centering
 \textsc{Accel}
\\
  \scalebox{0.45}{\includegraphics{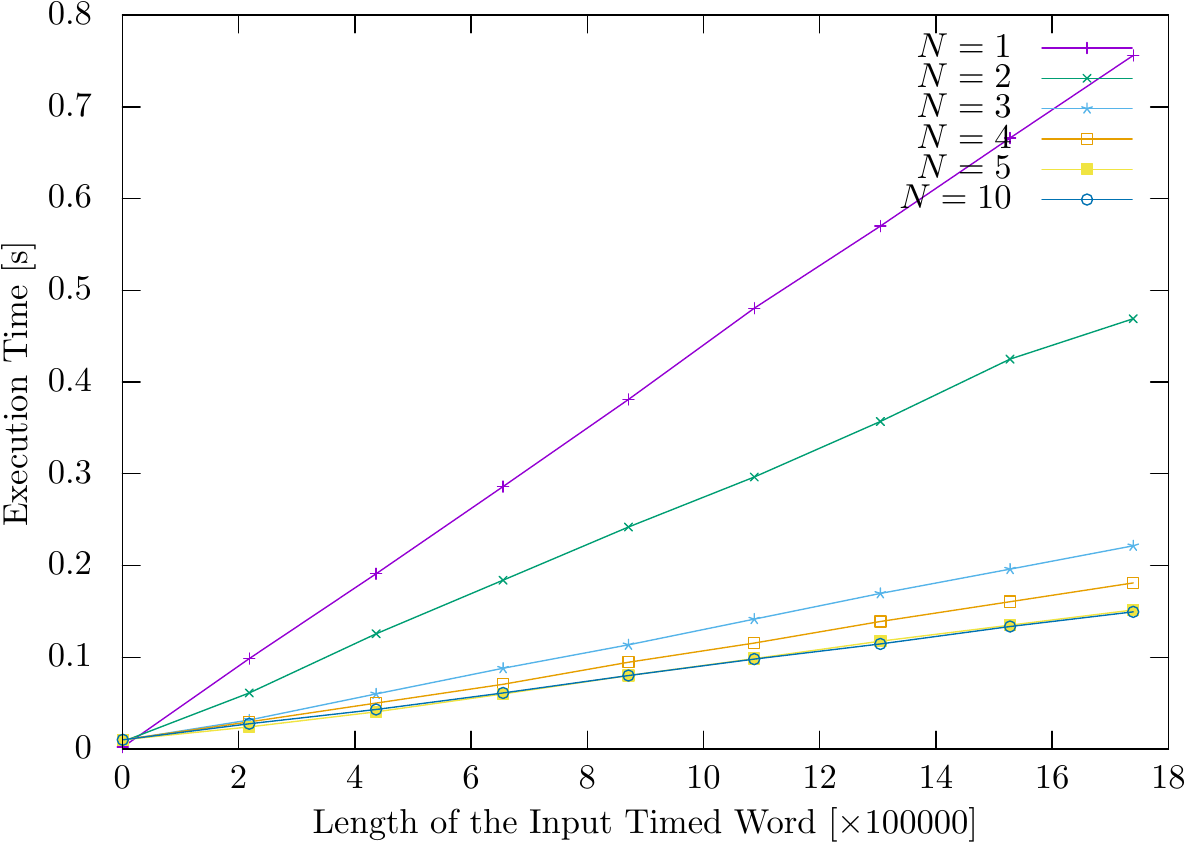}}
  \label{fig:execution_time_accel}
 \end{minipage}
  \caption{Execution time of our filter Moore machine (including time for construction of a filter)}
  \label{fig:execution_time}
\end{figure*}
\begin{figure*}[tbp]
 \begin{minipage}{.33\linewidth}
  \centering
 \textsc{Torque}
\\
  \scalebox{0.45}{ \includegraphics{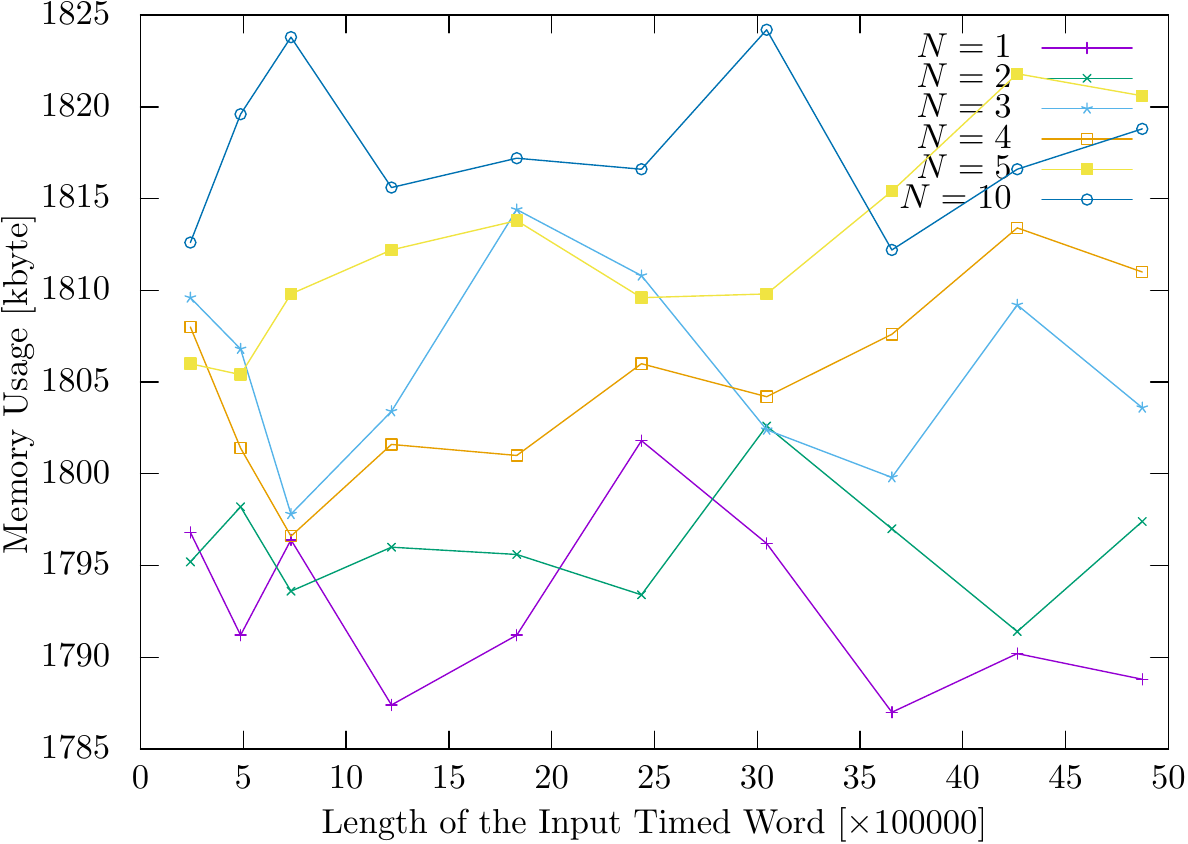}}
  \label{fig:memory_usage_torque}
 \end{minipage}
 \begin{minipage}{.33\linewidth}
  \centering
 \textsc{Gear}
\\
  \scalebox{0.45}{ \includegraphics{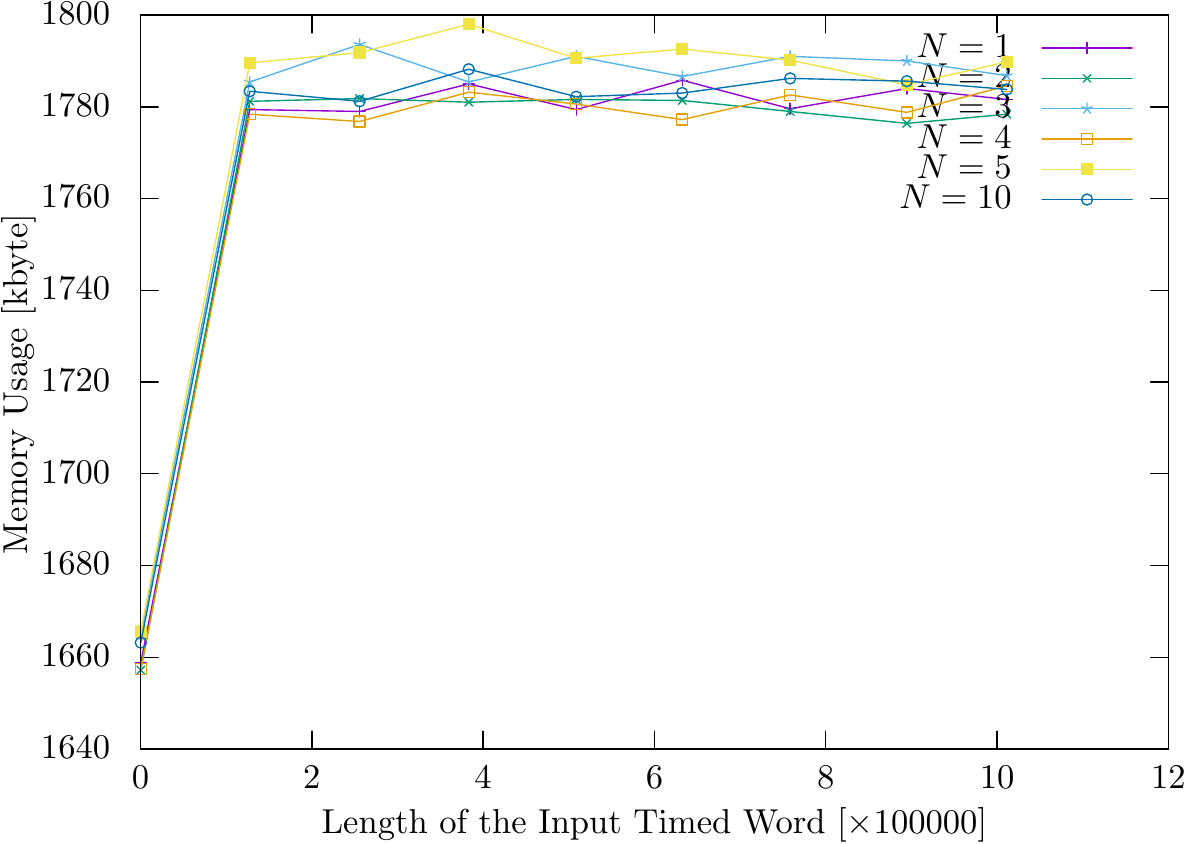}}
  \label{fig:memory_usage_gear}
 \end{minipage}
 \begin{minipage}{.33\linewidth} 
  \centering
 \textsc{Accel}
\\
  \scalebox{0.45}{ \includegraphics{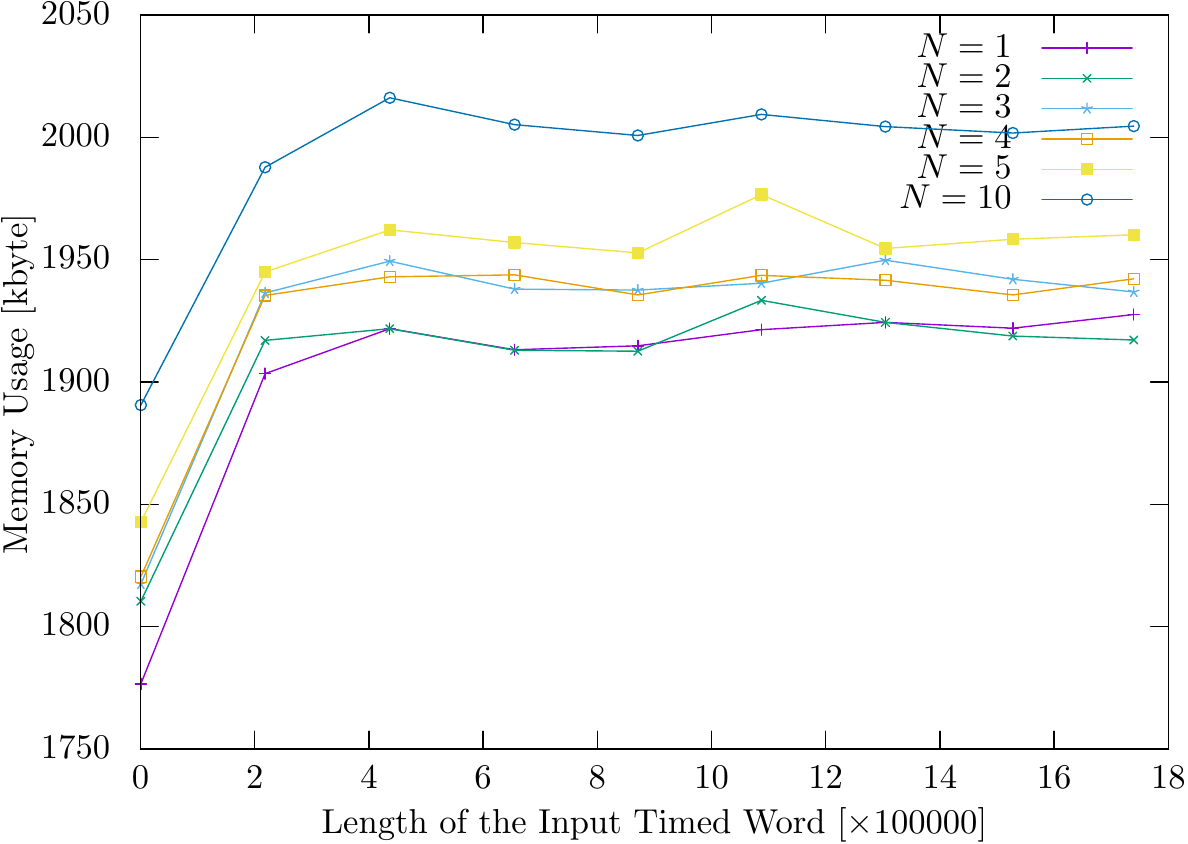}}
  \label{fig:memory_usage_accel}
 \end{minipage}
  \caption{Memory usage of our filter Moore machine}
  \label{fig:memory_usage}
\end{figure*}
\begin{figure*}[tbp]
 \begin{minipage}{.33\linewidth} 
  \centering
 \textsc{Torque}
\\
  \scalebox{0.45}{\includegraphics{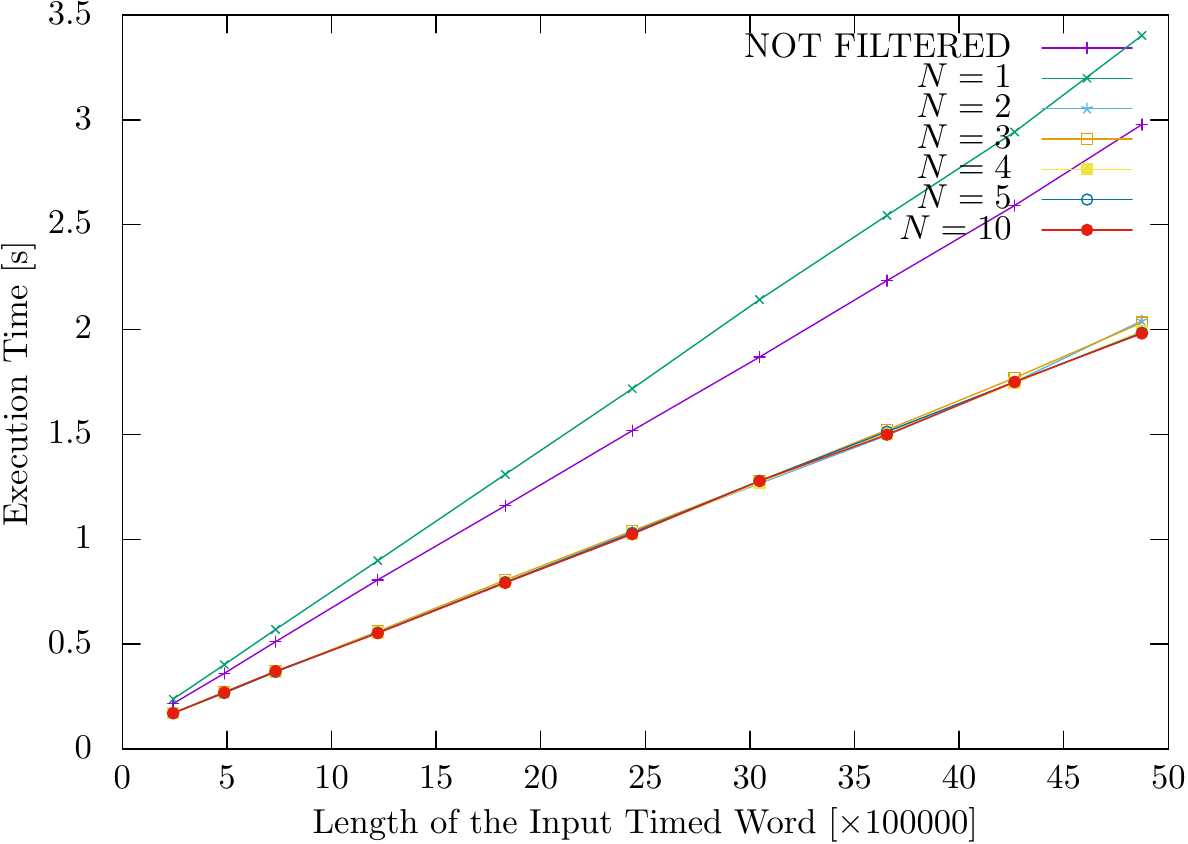}}
  \label{fig:filt_monaa_time_torque}
 \end{minipage}
 \begin{minipage}{.33\linewidth} 
  \centering
 \textsc{Gear}
\\
  \scalebox{0.45}{\includegraphics{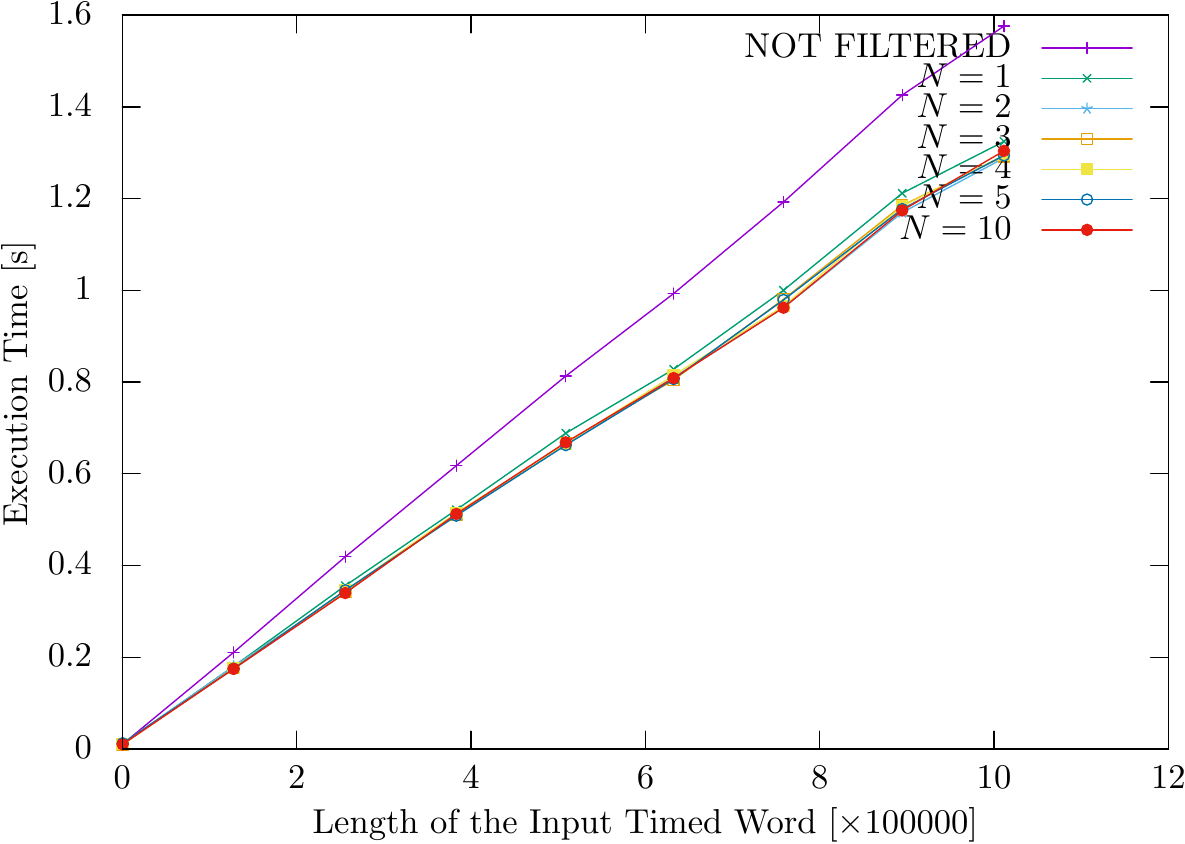}}
  \label{fig:filt_monaa_time_gear}
 \end{minipage}
 \begin{minipage}{.33\linewidth} 
  \centering
 \textsc{Accel}
\\
  \scalebox{0.45}{\includegraphics{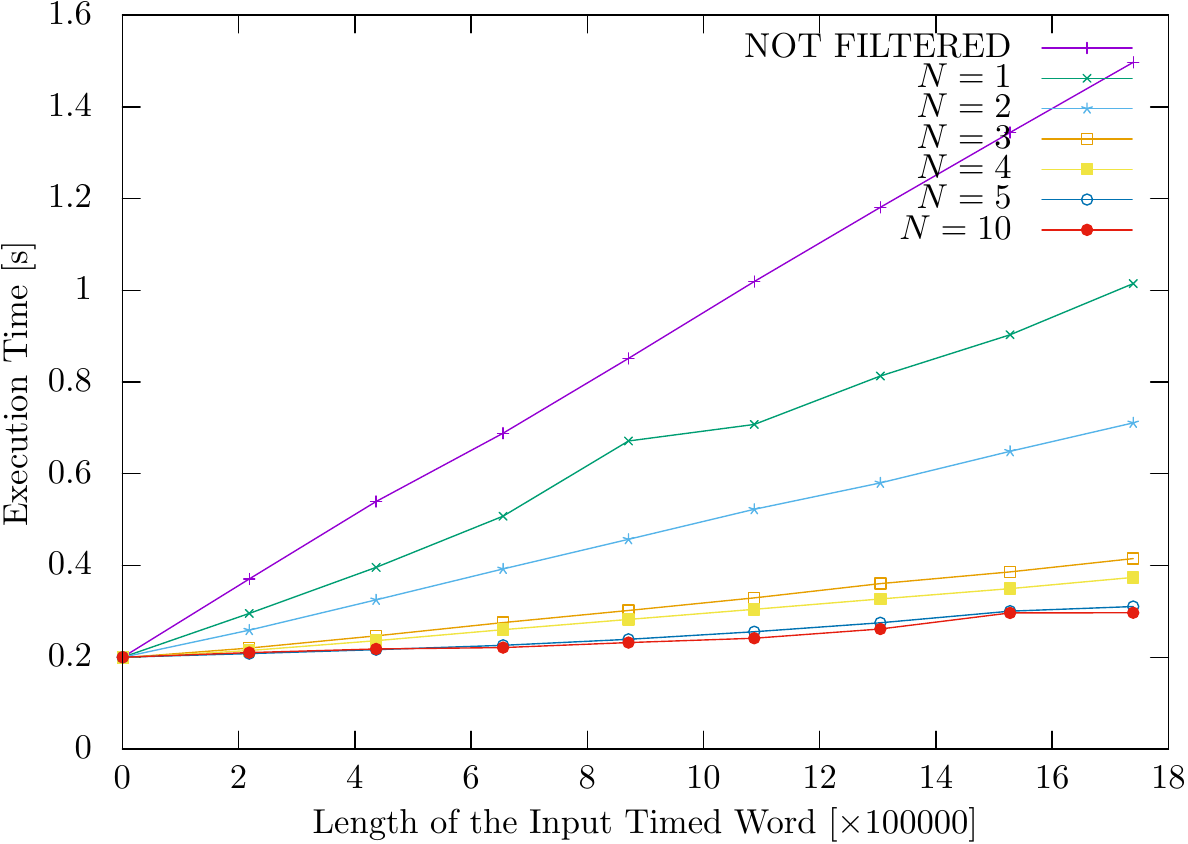}}
  \label{fig:filt_monaa_time_accel}
 \end{minipage}
  \caption{Execution time of the workflow in Fig.~\ref{fig:workflow}, with  our filter and monitoring by \monaa}
  \label{fig:filt_monaa_time}
\end{figure*}

\myparagraphdense{RQ1: Filtering Rate}
Fig.~\ref{fig:filtering_rate} shows
the length of the filtered timed words by our filter Moore machine for each
pattern timed automaton $\mathcal{A}$, buffer size $N$, and timed word
$w \in W$.


We observe that the filtered  word gets shorter as $N$ is bigger. This concurs with our theoretical consideration in Theorem~\ref{thm:nfa_monotonicity} (although the result is for the untimed setting). 
It seems that peak performance is achieved with a relatively small $N$, such as  $N = 10$. With $N=10$ the length of the filtered timed word is about $1/3$,
$1/2$, and $1/100$
of that of the original timed word in \textsc{Torque}, \textsc{Gear}, and
\textsc{Accel} respectively.
For \textsc{Accel} our filter masks many characters, because the sizes of the alphabet and the pattern timed automaton are relatively large.
This significant data 
reduction demonstrates the practical use of our filtering methodology in
 embedded usage scenarios (see Fig.~\ref{fig:hardwareArch}).

\myparagraphdense{RQ2: Speed and Memory Usage}
Fig.~\ref{fig:execution_time} and~\ref{fig:memory_usage} show
the execution time and memory usage of our filter Moore machine for each
pattern timed automaton $\mathcal{A}$, buffer size $N$, and timed word
$w \in W$.

In
Fig.~\ref{fig:execution_time},
we observe that the execution time is linear to the length of the
input word.
In Fig.~\ref{fig:memory_usage},
 the memory usage is more or less constant with respect to
the length of the
input word. These two results suggest that our filtering methodology is online capable.

The time for constructing a filter Moore machine is seen to be negligible: see the execution time for short input words in Fig.~\ref{fig:execution_time}.

On the effect of varying the buffer size $N$, for smaller $N$ the execution time was relatively large. This seems to be because fewer characters get masked, more characters are output, and this incurs  I/O cost. On memory usage, despite the worst-case result in Prop.~\ref{prop:filter_state_space} (exponential in $N$), the growth for bigger $N$ was moderate. This is because not all states from the powerset construction are reachable.



\myparagraphdense{RQ3: Acceleration of Timed Pattern Matching}
Fig.~\ref{fig:filt_monaa_time} shows the execution time of
the  workflow in Fig.~\ref{fig:workflow}, where the filter is given by our algorithm and  pattern matching is done by a recent tool \monaa~\cite{MONAACode}. 
More specifically, we
connect the standard output of the filter Moore machine to the standard
input of  \monaa by Unix pipeline. This way the
filter and \monaa run parallelly on different cores.

We observe that, via filtering, the overall
performance of timed pattern matching is  accelerated when $N$ is large enough (say $N=10$). For
\textsc{Torque} and \textsc{Gear}, the speedup is
$1.2$ times; for \textsc{Accel}, it is roughly
twice.
This acceleration suggests that our filtering methodology is potentially beneficial independent of the architecture assumption in Fig.~\ref{fig:hardwareArch}: when a log is enormous monitoring can take hours or days; by running a filter in parallel the time can be shortened. 

\myparagraphdense{RQ4: Precision} For the three examples \textsc{Torque}, \textsc{Gear}, and
\textsc{Accel}, and
 $N=10$, the ratios of non-masked
events contributing to actual matches were 0.34\%, 99\%, and 92\% respectively.
Therefore the precision differs greatly depending on patterns. It should also be noted that, even in the low-precision example (\textsc{Torque}), our filter successfully reduced the log size by approximately three times (Fig.~\ref{fig:filtering_rate}). 

 Most of the imprecision in the current timed setting is attributed to the laxness of one-clock determinization (Def.~\ref{def:oneClockDet}). 
For example,  the TA for  \textsc{Torque} (Fig.~\ref{fig:case1_pattern}) requires four consecutive occurrences of $\textrm{high}$ within one second, using the same clock $x$. 
The best overapproximation by one-clock determinization---in which all clock variables get reset after each transition---is to require an occurrence of $\textrm{high}$ in each of four consecutive time segments of length $\le 1$. This is a much looser requirement than the original, and explains the comparatively poorer precision in the example \textsc{Torque}. 

\myparagraphdense{RQ5: Responsiveness} For the three examples
\textsc{Torque}, \textsc{Gear}, \textsc{Accel}, and the buffer size $N =
10$, we calculated the average delay caused by our filter, which is
$\text{(the execution time)} / |w| \times N$. The results were 
$2.2\,\mathrm{\mu s}$, $3.1\,\mathrm{\mu s}$, and $0.91\,\mathrm{\mu s}$
respectively. Although these delays highly depend on the computation
power of the processor, the delays were tiny.  We conclude our filter
is responsive enough.

\auxproof{\myparagraph{RQ4. Performance Comparison with Timed Pattern Matching}
\begin{figure*}[tbp]
\begin{minipage}{.3\linewidth}
  \centering
 \scalebox{0.4}{\includegraphics{./figs/filt-vs-monaa-torque-standalone-trimmed.pdf}}
 \caption{Execution time of \monaa and our filter Moore machine (\textsc{Torque})}
 \label{fig:filt_vs_monaa_time_torque}
\end{minipage} 
\begin{minipage}{.3\linewidth}
 \centering
 \scalebox{0.4}{\includegraphics{./figs/filt-vs-monaa-gear-standalone-trimmed.pdf}}
 \caption{Execution time of \monaa and our filter Moore machine (\textsc{Gear})}
 \label{fig:filt_vs_monaa_time_gear}
\end{minipage} 
\begin{minipage}{.3\linewidth}
 \centering
 \scalebox{0.4}{\includegraphics{./figs/filt-vs-monaa-accel-standalone-trimmed.pdf}}
 \caption{Execution time of \monaa and our filter Moore machine (\textsc{Accel})}
 \label{fig:filt_vs_monaa_time_accel}
\end{minipage}
\end{figure*}

Fig.~\ref{fig:filt_vs_monaa_time_torque}--\ref{fig:filt_vs_monaa_time_accel}
are the execution time of our filter Moore machine and \monaa for each pattern timed automaton $\mathcal{A}$, buffer size $N$, and timed words $W$.

We observe that the execution time of filter Moore machine is about
0.5--0.8 of that of timed pattern matching when $N$ is large enough
(e.g., $N=10$).
It leads that our filter Moore machine requires less computation
resource and it is suitable for a weak processor, e.g., embedded system.
}

\section{Related Work}\label{sec:related}

Efficiency of pattern matching has been actively pursued in the fields of database~\cite{KandhanTP10,DBLP:conf/sigmod/MajumderRV08} and networking~\cite{LiuSLGF12,DBLP:conf/icn/ZhouXQL08,DBLP:journals/ijsn/ErdoganC07,DBLP:conf/ancs/YuCDLK06,DBLP:conf/uss/MeinersPNTL10,DBLP:conf/sigcomm/SmithEJK08}. In these fields,  issues in hardware architecture---such as speed gap between L1/L2 caches and main memory---constitute  situations similar to that of embedded monitors that we discussed in the above. 

 Many works in these application domains deal with (potentially multiple) strings as patterns. There the main source of inspiration is classic algorithms such as Boyer--Moore~\cite{DBLP:journals/cacm/BoyerM77}, Commentz-Walter~\cite{DBLP:conf/icalp/Commentz-Walter79}, and Aho--Corasick~\cite{DBLP:journals/cacm/AhoC75}. Examples are e.g.\ in~\cite{DBLP:journals/jea/SalmelaTK06}. Many algorithms for patterns given by regular expressions or automata (instead of strings) also rely on those string-matching techniques. See e.g.~\cite{KandhanTP10}.

In database and networking, pattern matching against regular expressions is mainly approached by application-specific heuristics that often take machine architecture into account. For example, in~\cite{LiuSLGF12,DBLP:conf/ancs/YuCDLK06} pattern regular expressions undergo application of some application-specific rewrite rules. The hierarchical structure of L1/L2 caches and main memory is exploited in~\cite{KandhanTP10,DBLP:conf/sigmod/MajumderRV08}, while the features of \emph{content-addressable memory} (CAM)---an alternative to RAM often used in network devices---are exploited in~\cite{DBLP:conf/uss/MeinersPNTL10}. The work~\cite{DBLP:conf/sigcomm/SmithEJK08} extends the formalism of DFA by auxiliary variables for the purpose of moderating state space inflation when multiple patterns are combined. 

Pre-filtering before actual pattern matching has been considered in the above lines of works~\cite{KandhanTP10,DBLP:journals/jea/SalmelaTK06,LiuSLGF12}. The principal difference between those works and ours is that their filters produce \emph{match candidates}, data that include explicit indices for potential matches. For this reason, the second step of the workflow (what we call pattern matching) is called \emph{verification} in those papers. In contrast, our filter only \emph{masks} the input word. This is because our goal---motivated by embedded applications---is not only in the matching speed but also in reducing the amount of data passed from the sensor to the pattern matching unit. This choice enables us to use Moore machines, too, which can be readily implemented on FPGA and ASIC.

Monitoring over a real-time temporal logic
property~\cite{DBLP:conf/formats/MalerN04,DBLP:journals/jacm/BasinKMZ15} and
timed pattern matching~\cite{DBLP:conf/formats/UlusFAM14,DBLP:conf/formats/WagaAH16} are relatively new topics. They have been mainly pursued in the context of cyber-physical systems, although its applications in database and networking are very likely. In timed pattern matching, specifications can be given by timed automata (as we do in the current work), \emph{timed regular expressions}~\cite{DBLP:journals/jacm/AsarinCM02} and \emph{metric temporal logic} formulas~\cite{DBLP:conf/focs/AlurH92}. Algorithms for timed pattern matching against specifications in these formalisms have been actively pursued 
in~\cite{DBLP:conf/formats/UlusFAM14,DBLP:conf/tacas/UlusFAM16,DBLP:conf/cav/Ulus17,DBLP:conf/formats/BakhirkinFMU17,DBLP:conf/formats/AsarinMNU17}. 
Besides, 
the line of work~\cite{DBLP:conf/formats/WagaAH16,DBLP:conf/formats/WagaHS17,MONAACode}
has pursued acceleration of timed pattern matching, by combining \emph{shift table} techniques (like in Boyer--Moore~\cite{DBLP:journals/cacm/BoyerM77} and Franek--Jennings--Smyth~\cite{DBLP:journals/jda/FranekJS07}) and timed automata~\cite{DBLP:conf/formats/WagaAH16,DBLP:conf/formats/WagaHS17}. This idea of automata-theoretic shift tables is pioneered in~\cite{DBLP:journals/scp/WatsonW03}.

\section{Conclusions and Future Work}\label{sec:concl}
Motivated by the recent rise of monitoring needs in embedded applications, we presented the construction of filtering Moore machines for (untimed and timed) pattern matching. The construction is automata-theoretic, realizing filters as Moore machines. 

We will pursue embedded implementation of the proposed technique. In particular, we will investigate hardware acceleration by FPGA or ASIC, exploiting  that our filters are Moore machines~\cite{Reese06}. 


On the theoretical side, we plan to further investigate the relationship among: automata theory, lightweight formal methods for cyber-physical and embedded systems, and other fields like database and networking. The papers in~\S{}\ref{sec:related} seem to suggest that there are many techniques waiting for being exported from one field to another. 

Besides, a generalization to a distributed setting was suggested by anonymous reviewers. It is often the case with the actual embedded systems, and we believe it will be interesting future work.

\bibliographystyle{IEEEtran}
\bibliography{dblp_refs}

\begin{IEEEbiography}[{\includegraphics[width=1in,height=1.25in,clip,keepaspectratio]{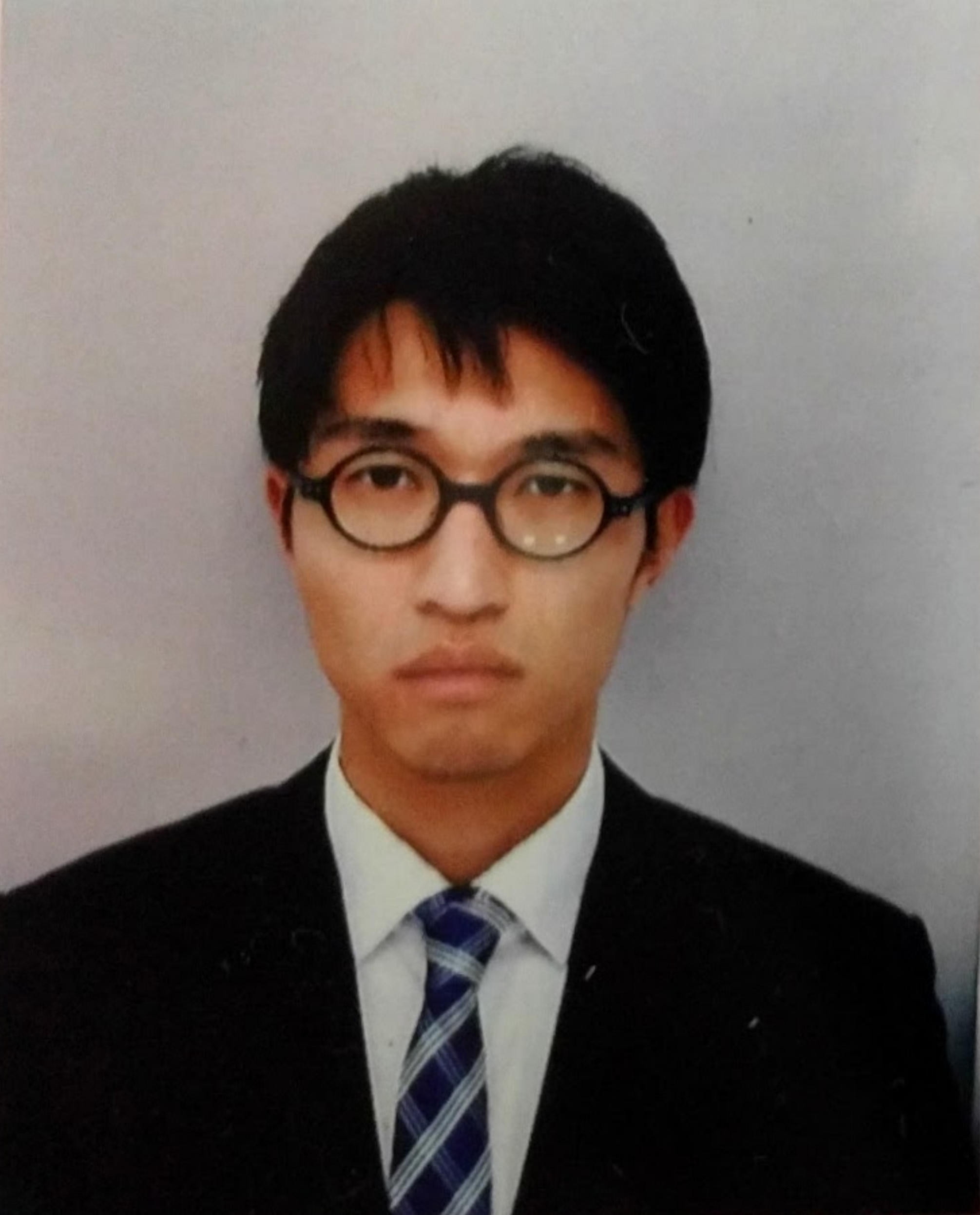}}]{Masaki Waga}
 received the B.S and M.S degrees in computer science from the
 University of Tokyo, Tokyo, Japan, in 2016 and 2018, respectively.
 He is currently pursuing the Ph.D degree in informatics with
 SOKENDAI, Kanagawa, Japan.

 His current research interests include runtime verification of
 cyber-physical systems and theory of timed automata.
\end{IEEEbiography}

\begin{IEEEbiography}[{\includegraphics[width=1in,height=1.25in,clip,keepaspectratio]{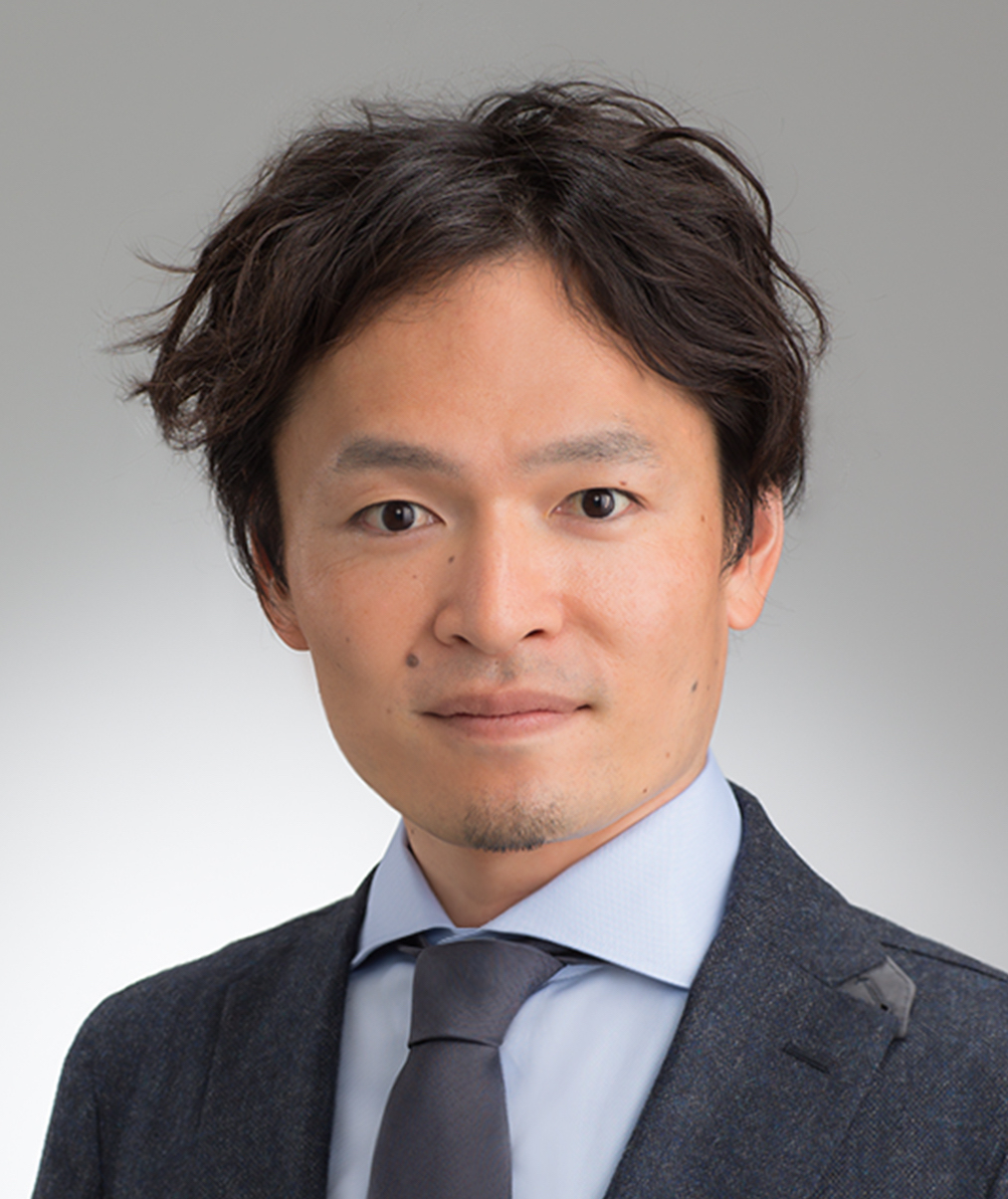}}]{Ichiro Hasuo} is an Associate Professor at National Institute of Informatics (NII), Tokyo, Japan, where he also leads JST ERATO Metamathematics for Systems Design Project. He received the BSc and MSc degrees from the University of Tokyo (2002) and Tokyo Institute of Technology (2004), respectively, and received the PhD degree (cum laude) from Radboud University Nijmegen, the Netherlands (2008). Prior to joining NII,  he was an Assistant Professor at RIMS, Kyoto University, and  a Lecturer and  an Associate Professor at the University of Tokyo. 
His research interests are in mathematical (logical, algebraic and categorical) structures in  formal methods; abstraction and generalization of deductive and automata-theoretic verification techniques; and their application to cyber-physical systems. 
\end{IEEEbiography}







 \newpage
 \appendix

 \section{Omitted Proofs}

 \begin{mynotation}
 [$\stuck{w}$]
 \label{notation:stuck}
 Let $\mathcal{A} =(\Sigma, S, s_0, E,S_F)$ be an NFA, and 
 $w=a_{1}\dotsc a_{n}\in\Sigma^{*}$ be a word. The fact that reading $w$ in $\mathcal{A}$ leads to no active states, that is, 
 \begin{math}
 \bigl\{\,s\in S\,\mid\,\exists s_{1},\dotsc, s_{n-1}\in S.\, 
 (s_{0}
 \xrightarrow{a_{1}}s_{1}
 \xrightarrow{a_{2}}\dotsc
 \xrightarrow{a_{n-1}}s_{n-1}
 \xrightarrow{a_{n}}s \text{ in $\mathcal{A}$})
 \,\bigr\} =\emptyset
 \end{math},
 is denoted by $\stuck{w}$.
 \end{mynotation}

 \begin{lemma}
 \label{lemma:filtered_character}
 Let $\mathcal{A}$ be an NFA $\mathcal{A} = (\Sigma, S, s_0, E,S_F)$, and
 $N$ be a positive integer. Let 
 $\mathcal{M}_{\mathcal{A},N}=(\Sigma_{\bot},\Sigma_{\bot},Q,q_0,\Delta,\Lambda)$
 be the  Moore machine in Def.~\ref{def:filterConstructionUntimed}. 
 of $\mathcal{A}$,
 $w=a_1 a_2 \dots a_n$ be a word over $\Sigma$, and $\bot^{N} w'$ be the
 output word of $\mathcal{M}_{\mathcal{A},N}$ over the input word
 $w \bot^{N}$ where $w' = b_1 b_2 \dots b_n$.
 For any index $k$ of $w$, we have $b_k=\bot$ if and only if 
 for any $i\in[1,k]$ and $k'\in[k,k+N-1]$, 
 we have either $(k'-i)\bmod N < k' - k$ or 
 $w|_{[i,k']}\not\in L(\mathcal{A})$, and 
 for any $k'\in[k,k+N-1]$ and $m \in \Zp$, we have
 $\stuck{w|_{[k'-mN+1,k']}}$.
 \end{lemma}

 \begin{proof}
 Let
 $(\mathcal{S}_0,\overline{a}_0,\overline{l}_0),(\mathcal{S}_{1},\overline{a}_{1},\overline{l}_{1}),\dots,(\mathcal{S}_{|w|+N},\overline{a}_{|w|+N},\overline{l}_{|w|+N})\in Q^*$
 be the run of $\mathcal{M}$ over $w  \bot^{N}$.
 By the definition of $\Lambda$, we have $b_k = \bot$ if and only if 
 we have $l_{k+N,1}=\mask$, which holds if and only if we have
 \begin{equation}
  \label{eq:filter_character_1}
  \forall k'\in[k,k+N-1], (s,n)\in\mathcal{S}_{k'}.\, (s\not\in S_F \lor n \leq k'-k) \land n\neq N
 \end{equation}
 by the definition of $\Delta$.
 The first part 
 \begin{math}
  \forall k'\in[k,k+N-1], (s,n)\in\mathcal{S}_{k'}.\, s\not\in S_F \lor n < k'-k
 \end{math}
 of (\ref{eq:filter_character_1}) holds if and only if for any
 $i\in [1,k]$, we have either $(k'-i)\bmod N < k' - k$ or $w|_{[i,k']}\not\in L(\mathcal{A})$.
 The second part
 \begin{math}
  \forall k'\in[k,k+N-1], (s,n)\in\mathcal{S}_{k'}.\, n\neq N
 \end{math}
 holds if and only if
 for any $k'\in[k,k+N-1]$ and $m \in \Zp$, we have
 $\stuck{w|_{[k'-mN+1,k']}}$.
 \end{proof}

 \subsection{Proof of Thm.~\ref{thm:nfa_soundness}}

 \begin{proof}
 Let $(i,j)$ be a pair of indices of $w$ satisfying 
 $w|_{[i,j]} \in L(\mathcal{A})$ and 
 $s_i,s_{i+1},\dots,s_{j}\in S^*$ be a
 sequence of states such that we have $(s_0,a_i,s_i) \in E$, 
 $s_j\in S_F$, and for each $k \in [i+1,j]$, $(s_{k-1},a_k,s_k)\in E$
 holds.
 Let
 $(\mathcal{S}_0,\overline{a}_0,\overline{l}_0),(\mathcal{S}_{1},\overline{a}_{1},\overline{l}_{1}),\dots,(\mathcal{S}_{n+N},\overline{a}_{n+N},\overline{l}_{n+N})\in
 Q^*$ be the run of
 $\mathcal{M}_{\mathcal{A},N}$ over $w \bot^{N}$.
 By the definition of the transition function $\Delta$ of
 $\mathcal{M}_{\mathcal{A},N}$, for each
 $k \in [i,j]$, we have $(s_{k},n_k)\in\mathcal{S}_{k}$ where 
 $n_k \in \mathbb{Z}/N\mathbb{Z}$. 
 Thus for each $m\in \Znn$ satisfying $i +N m \leq j$, we have
 $\overline{a}_{i+Nm+N-1} = (a_{i+Nm},a_{i+Nm+1},\cdots,a_{i+Nm+N-1})$
 and
 $\overline{l}_{i+Nm+N-1} = (\pass,\pass,\dots,\pass)$.
 Also, because of $s_j \in S_F$, we have 
 $\overline{a}_{j-N+1} = (a_{j-N+1},a_{j-N+1},\cdots,a_{j})$
 and for each $m\in [N-((j-i)\bmod N),N]$ we have 
 ${l}_{j-N+1,m} = \pass$.
 Thus, for any index $k \in [i,j]$ of $w$, we have
 $a_k = b_k$.
 \end{proof}

 \subsection{Proof of Thm.~\ref{thm:nfa_completeness}}

 \begin{proof}
 Since $\max\{|w|\mid w \in L(\mathcal{A})\} < \infty$, there is an NFA
 $\mathcal{A}'$ such that $L(\mathcal{A})=L(\mathcal{A}')$ holds and
 there is no transition from any accepting states.
 Let $\mathcal{A}'$ be such an NFA and
 $(\mathcal{S}_0,\overline{a}_0,\overline{l}_0),(\mathcal{S}_{1},\overline{a}_{1},\overline{l}_{1}),\dots,(\mathcal{S}_{|w|+N},\overline{a}_{|w|+N},\overline{l}_{|w|+N})\in Q^*$
 be the run of $\mathcal{M}_{\mathcal{A},N}$ over $w  \bot^{N}$.
 Since the length of a run of $\mathcal{A}'$ is not more than $N+1$ and
 the run is accepting if its length is $N+1$, $a_k=b_k$ holds if and only if there
 exist $i \in [1,k]$ and $j\in[k,k+N-1]$, we have 
 $w|_{[i,j]} \in L(\mathcal{A})$, by Lemma~\ref{lemma:filtered_character}.
 \end{proof}

 \auxproof{
 \subsection{Proof of Thm.~\ref{thm:nfa_substring_window}}

 \begin{proof}
 Without loss of generality, we assume that for any state $s \in S$ of $\mathcal{A}$, there is at
 least one accepting state $s_f$ reachable from $s$.
 If
 \begin{math}
 \bigcup_{d\in [0,M]}(\Sigma^* w|_{[k-d,k-d+M-1]} \Sigma^*)\cap
 L(\mathcal{A})=\emptyset  
 \end{math}
 holds, we have 
 \begin{equation}
 \label{eq:substr_1}
 \bigcup_{i\in [k-M+1,k],j\in [k,k+M-1],j-i\geq M-1}
 (\Sigma^* w|_{[i,j]} \Sigma^*) \cap
 L(\mathcal{A})=\emptyset \enspace . 
 \end{equation}
 Equation (\ref{eq:substr_1}) implies that
 for any $i\in [1,k],j\in [k,k+N-1]$ satisfying $j-i\geq M-1$, we have
 \begin{math}
 w|_{[i,j]}\not\in L(\mathcal{A})
 \end{math}.
 Since $j-i < M-1$ leads $w|_{[i,j]}\not\in L(\mathcal{A})$, 
 for any $i\in [1,k],j\in [k,k+N-1]$, we have
 \begin{math}
 w|_{[i,j]}\not\in L(\mathcal{A})
 \end{math}.
 Besides, equation (\ref{eq:substr_1}) implies that
 for any $j\in [k,k+N-1]$ and $m\in \Zp$, we have
 \begin{math}
 w|_{[j-mN+1,j]} \Sigma^* \cap L(\mathcal{A}) =\emptyset
 \end{math}.
 Since each state $s\in S$ of $\mathcal{A}$ has at least one accepting
 state reachable from $s$, 
 for any $j\in[k,k+N-1]$ and $m \in \Zp$, we have 
 \begin{math}
 \stuck{w|_{[j-mN+1,j]}}
 \end{math}.
 By Lemma~\ref{lemma:filtered_character}, we have $b_k=\bot$.
 \end{proof}
 }

 \subsection{Proof of Thm.~\ref{thm:nfa_monotonicity}}

 \begin{proof}
 By Lemma~\ref{lemma:filtered_character}, if we have
 $b^{(N)}_{k}=\bot$,
 we have 
 \begin{equation}
 \label{eq:monotinic_1}
 \forall i\in[1,k],j\in[k,k+N-1].\,w|_{[i,j]} \not\in L(\mathcal{A}) \lor
 (j-i) \bmod N < j-k
 \end{equation}
 and
 \begin{equation}
 \label{eq:monotinic_2}
 \forall j\in[k,k+N-1],m \in \Zp.\, 
 \stuck{w|_{[j-mN+1,j]}}
 \enspace .
 \end{equation}
 Because of (\ref{eq:monotinic_1}), we have either
 \begin{equation}
 \label{eq:monotinic_3}
 \forall i\in[1,k],j\in[k,k+N-1].\,w|_{[i,j]} \not\in L(\mathcal{A}) \lor
 (j-i) \bmod nN < j-k 
 \end{equation}
 or
 \begin{equation}
 \label{eq:monotinic_4}
  \forall i\in[1,k],j\in[k,k+N-1].\,w|_{[i,j]} \not\in L(\mathcal{A}) \lor
 \exists p \in [1,n-1].\, pN \leq j-i-qnN < Np + j-k
 \end{equation} where $q$ is the positive integer satisfying 
 \begin{math}
 0 \leq j-i-qnN < nN
 \end{math}.
 Because of
 \begin{math}
 pN \leq j-i-qnN < Np + j-k \iff k \leq j-i-qnN-pN+k < j
 \end{math}
 and
 \begin{math}
 pN \leq j-i-qnN < Np + j-k \leq pN + N
 \end{math}, (\ref{eq:monotinic_4}) implies
 \begin{math}
 \forall i\in[1,k],j\in[k,k+N-1].\,w|_{[i,j]} \not\in L(\mathcal{A}) \lor
 \exists q \in \Zp.\, k\leq j-i-qN +k < j
 \end{math}, which is equivalent to
 \begin{math}
 \forall i\in[1,k],j\in[k,k+N-1].\,w|_{[i,j]} \not\in L(\mathcal{A}) \lor
 \exists q \in \Zp.\, k< i+qN \leq j
 \end{math}.
 Because of (\ref{eq:monotinic_2}), if there exists $q \in \Zp$
 satisfying $k < i+qN \leq j$, we have 
 $\stuck{w|_{[i,i+qN]}\cdot w|_{[i+qN+1,j]}}$
 Thus, (\ref{eq:monotinic_4}) implies $w|_{[i,j]}\not\in L(\mathcal{A})$ and
 we have (\ref{eq:monotinic_3}).
 Besides, by (\ref{eq:monotinic_2}), for any $j'\in[k,k+nN-1]$ and $m' \in \Zp$,
 we have
 \begin{math}
 \stuck{w|_{[j'-mN+1,j'']}}
 \end{math}
 where $j'' = k + ((j'-k) \bmod N)$, which leads
 \begin{math}
 \stuck{w|_{[j'-mN+1,j'']}\cdot w|_{[j''+1,j']}}
 \end{math}.

 Because of (\ref{eq:monotinic_2}), for any $i\in[1,k]$ there exists
 $k'\in[k,k+N-1]$ such that we have 
 \begin{math}
 \stuck{w|_{[i,k']}}
 \end{math}.
 Thus, for any $i\in[1,k]$ and $j\in[k+N,k+nN-1]$, we have 
 \begin{math}
 w|_{[i,j]} \not\in L(\mathcal{A})
 \end{math}.
 Thus, we have both
 \begin{displaymath}
 \forall i\in[1,k],j\in[k,k+nN-1].\,w|_{[i,j]} \not\in L(\mathcal{A}) \lor (j-i) \bmod nN < j-k 
 \end{displaymath}
 and
 \begin{displaymath}
 \forall j\in[k,k+nN-1],m \in \Zp.\, 
 \stuck{w|_{[j-mnN+1,j]}}
 \enspace ,
 \end{displaymath} 
 and we have $b^{(nN)}_{k}=\bot$.
 \end{proof}

 \subsection{Proof of Thm.~\ref{thm:timed_soundness}}

 \begin{proof}
 If we have $w(i,j)-\tau_{i-1} \in L(\mathcal{A})$, there is an accepting run 
 $((s_0,s_i,s_{i+1},\dots,s_{j}),(\nu_0,\nu_i,\nu_{i+1},\dots,\nu_{j}))$
 of $\mathcal{A}$ over $w(i,j)-\tau_{i-1}$.
 Let $\overline{\nu} = (\nu_0,\nu_i,\nu_{i+1},\dots,\nu_{j})$.
 By the definition of $\mathcal{A}^{N\text{-}\mathrm{ctr}}$, there is an accepting run
 $(((s_0,0),(s_i,n_{i}),(s_{i+1},n_{i+1}),\dots,(s_{j},n_{j})),\overline{\nu})$
 of $\mathcal{A}^{N\text{-}\mathrm{ctr}}$ over $w(i,j)-\tau_{i-1}$ where for
 any $k\in[i,j]$, $n_k = (k-i \bmod N)+1$.
 Let
 $((\mathcal{S}_0,\mathcal{S}_1,\dots,\mathcal{S}_{j}),(\nu'_0,\nu'_1,\dots,\nu'_{j}))$
 be the run of $\mathcal{A}^{N\text{-}\mathrm{ctr\text{-}d}}$ over $w(1,j)$.
 By Prop.~\ref{prop:propertiesOfOneClockDeterminization}, 
 for any $k\in[i,j]$ we have $((s_k,n_k),Z_k) \in \mathcal{S}_k$ where
 $s_k = (k-i \bmod N)+1$ and $s_j \in S_F$.
 Let
 $((\mathcal{S}_0,\overline{l}_0),(\mathcal{S}_1,\overline{l}_1),\dots,(\mathcal{S}_{n+N},\overline{l}_{n+N}))$
 be the run of $\mathcal{M}_{\mathcal{A},N}$ over $w (\bot,\tau_n)^{N}$.
 By the definition of the filter Moore machine $\mathcal{M}_{\mathcal{A},N}$,
 for any $m \in \Znn$ satisfying $i + Nm-1 \leq j$, we have
 $b_{i + Nm-1} = N$ and
 $\overline{l}_{i+Nm-1}=\pass^N$.
 Also, because of $\mathcal{S}_{j} \in S'_F$, 
 we have
 $\overline{l}_{j} =
 l_2,l_3,\dots,l_{N-\psi(\mathcal{S}_{j})+1},\overbrace{\pass,\dotsc,\pass}^{\psi(\mathcal{S}_{j})}$
 where for any $k\in[2,N-\psi(\mathcal{S}_{j})+1]$, $l_k$ is either
 $\pass$ or $\mask$.
 Thus, for any $k\in[i,j]$, we have $b_k = \pass$.
 \end{proof}

 \section{Filtering with On-The-Fly Determinization}
 \label{appendix:nfa_filtering}

 See Alg.~\ref{alg:nfa_filtering}.

 \begin{algorithm*}[tb]
 \caption{A Filtering Algorithm for Pattern Matching with On-The-Fly Determinization}
 \label{alg:nfa_filtering}
 \scalebox{0.9}{
 \parbox{\linewidth}{
 \begin{algorithmic}[1]
  \Require Word $\str \in \Sigma^*$, NFA 
  $\mathcal{A} = (\Sigma, S, s_0, E,S_F)$, $N \in \Zp$
  \Ensure $\bot^{N-1} \cdot w' \in \Sigma_{\bot}$ is the output word of
  $\mathcal{M}_{\mathcal{A},N}$ given the input word $w\cdot \bot^{N-1}$
  \State $\currSN \gets \{(s_0,0)\};\overline{a} \gets \bot^N;\overline{l} \gets \mask^N$
  \For{$i = 1; i \leq |\str|; i \gets i + 1$}
  \State \scalebox{0.95}{$\nextSN \gets \{(s_0,0)\}
  \cup \{(s',n+1\bmod N) \mid \exists (s,n) \in \currSN, s' \in S.\, (s,w(i),s') \in E \}$}\label{line:nfa_filtering_begin_waiting}
  \State $\overline{a} \gets {a}_2,a_3,\dots,a_N,w_i$
  \If{$\exists (s,N) \in \nextSN$}
  \State $\overline{l} \gets \pass^N$
  \ElsIf{$\exists (s,n) \in \nextSN.\, s \in S_F$}\label{line:nfa_filtering_begin_sending}
  \State $\mathit{maxN} \gets \max\{n\mid (s,n) \in \nextSN, s \in S_F\}$
  \State $\overline{l} \gets {l}_2,l_3,\dots,l_{N-\mathit{maxN}+1},\pass^{\mathit{maxN}}$
  \Else
  \State $\overline{l} \gets {l}_2,l_3,\dots,l_N,\mask$
  \EndIf\label{line:nfa_filtering_end_sending}
  \If{$l_1 = \pass$}
  \State $w'\gets w'\cdot a_1$
  \Else
  \State $w'\gets w'\cdot \bot$
  \EndIf
  \State $\currSN \gets \nextSN$
  \EndFor
  \For{$i = 1;i \leq N-1;i\gets i+1$}
  \If {$l_i = \pass$}
  \State $w'\gets w'\cdot a_i$
  \Else
  \State $w'\gets w'\cdot \bot$
  \EndIf
  \EndFor
 \end{algorithmic}}}
 \end{algorithm*}

\section{A Sketch of One-Clock Determinization}
\label{appendix:one_clock_determinization}

We shall now sketch the one-clock determinization.
We can split the one-clock determinization into two steps: 1)
overapproximation of a timed automaton by a \emph{real-time automaton (RTA)}; and 2)
determinization of the resulting real-time automaton. 

\begin{algorithm*}[tb]
 \caption{An over-approximation of a timed automaton by a real-time automaton.}
 \label{alg:approx_ta_by_rta}
 \scalebox{0.9}{
 \parbox{\linewidth}{
 \begin{algorithmic}[1]
  \Require a timed automaton $\mathcal{A}=(\Sigma,S,s_0,S_F,C,E)$
  \Ensure a real-time automaton $\mathcal{A}^{\mathrm{rt}}=(\Sigma,S^{\mathrm{rt}},s^{\mathrm{rt}}_0,S^{\mathrm{rt}}_F,\{y\},E^{\mathrm{rt}})$
  \State $s^{\mathrm{rt}}_0 \gets(s_0,\{\mathbf{0}+t\mid t \in \Rnn\}); \quad S^{\mathrm{rt}}\gets\{s^{\mathrm{rt}}_0\};\quad \currS \gets \{s^{\mathrm{rt}}_0\}$
  \While{$\currS \ne \emptyset$}
  \State $\nextS \gets \emptyset$
  \For{$(s,Z) \in \currS$}   \Comment{$Z$ is a zone in $(\R_{\geq 0})^{C\amalg\{y\}}$}
  \For{$(s,a,\delta,\lambda,s') \in E$}
  \State $Z'\gets \{\nu \mid \nu \in Z, \nu\models\delta\}$\label{alg_line:approx_start_dbm_guard}
  \If{$Z' \ne \emptyset$} \label{alg_line:emptycheck}
  \State $I\gets \{\nu|_{\{y\}} \mid \nu \in Z'\}$ \label{alg_line:approx_end_dbm_interval}
  \State $Z''\gets\{\nu[x \mapsto 0]_{x\in (\lambda\amalg\{y\})} + t \mid
  \nu \in Z', t \in \Rnn\}$\label{alg_line:approx_end_dbm_guard} 
  \label{alg_line:reset_and_time_elapse}
  \State $E^{\mathrm{rt}}\gets E^{\mathrm{rt}}\cup \{((s,Z),a,y \in I,\{y\},(s',Z''))\}$
   \label{alg_line:adding_transition}
  \If{$(s',Z'') \not\in S^{\mathrm{rt}}$}
  \State $\nextS \gets \nextS \cup \{(s',Z'')\}$
  \EndIf
  \EndIf
  \EndFor
  \EndFor
  \State $\currS \gets \nextS;\quad S^{\mathrm{rt}}\gets S^{\mathrm{rt}} \cup\nextS$
  \EndWhile
  \State $S^{\mathrm{rt}}_F \gets \{(s,Z) \in S^{\mathrm{rt}}\mid s \in S_F\}$
 \end{algorithmic}}}
\end{algorithm*}

\myparagraph{Overapproximation of TA by RTA}
\emph{Real-time automata (RTA)}~\cite{DBLP:conf/stacs/Dima00} form a subclass of TA. Intuitively, in
an RTA the enabledness of a transition is determined solely by the dwell
time at the source state of the transition. This amounts, in technical terms, to the condition that an RTA has only one clock
variable, and the clock is reset after each transition.

Algorithm~\ref{alg:approx_ta_by_rta}, that gives an RTA $\mathcal{A}^{\mathrm{rt}}$ such that $L(\mathcal{A})\subseteq L(\mathcal{A}^{\mathrm{rt}})$, is essentially
                     the construction described
                     in~\cite[\S{}5.3.4]{DBLP:journals/fmsd/KrichenT09}. The
                     difference is as follows. In
                     Line~\ref{alg_line:approx_end_dbm_interval} we use zones
                     and their representation by \emph{difference bound
                     matrices (DBM)} for constructing new guards, while
                     in~\cite[\S{}5.3.4]{DBLP:journals/fmsd/KrichenT09}
                     a region-like construction is suggested  (the latter should be more
                     expensive) . We  use DBM also in
                     Lines~\ref{alg_line:approx_start_dbm_guard}, \ref{alg_line:emptycheck} and~\ref{alg_line:approx_end_dbm_guard};
                     this is the same as
                     in~\cite[\S{}5.3.4]{DBLP:journals/fmsd/KrichenT09}.

In Line~\ref{alg_line:reset_and_time_elapse}, we take a 
 clock valuation $\nu\in Z'$ and first reset all the clocks in $\lambda$ and additionally $y$, and then take $t$-shift for  arbitrary $t$.
We note that each endpoint of the interval $I$ is an integer or $+\infty$, and
hence the condition
``$y \in I$'' (e.g.\ in Line~\ref{alg_line:adding_transition}) can be represented by a constraint in $\Phi(\{y\})$. 

Termination of Algorithm~\ref{alg:approx_ta_by_rta} can be established much like the termination of usual zone automata constructions. 

\myparagraph{Determinization of RTA}
Determinization of RTA is possible unlike general TA. This is because each transition resets the unique clock $y$, and therefore the enabledness of a transition is determined locally.

For an RTA $\mathcal{A}=(\Sigma, S, s_0, S_F, \{y\},E)$, its
determinization
$\mathcal{A}^{\mathrm{d}}=(\Sigma, \mathcal{P}(S), \{s_0\}, S_F^{\mathrm{d}}, \{y\},E^{\mathrm{d}})$ can be defined as follows. 
 The accepting states are defined by $S_F^{\mathrm{d}} =\{\mathcal{S}\in \mathcal{P}(S)\mid \mathcal{S}\cap S_{F}
\neq\emptyset
\}$, as usual. On  transitions,
for any $\mathcal{S}\in \mathcal{P}(S)$ and $a \in \Sigma$,
let $\mathbb{I}_{\mathcal{S},a} = \{I_1,I_2,\dots,I_n\}$ be the coarsest
partition of $\Rnn$ such that for any $i\in \{1,2,\dotsc,n\}$ and
$(s,a,y\in I,\{y\},s') \in E$,
we have either $I_i \subseteq I$ or $I_i \cap I = \emptyset$.
 This is used in the definition of the transition set $  E^{\mathrm{d}}$:
\begin{multline*}
  E^{\mathrm{d}} = \bigl\{\;(\mathcal{S},a,y \in I',\mathcal{S}') \;\big|\; I' \in
 \mathbb{I}_{\mathcal{S},a},\\ \mathcal{S}'=\{s'\in S\mid
 \exists s \in \mathcal{S}.\,\exists(s,a,y\in I'',\{y\},s')\in E.\enspace I'\subseteq I''\}\;\bigr\}
 \enspace .
\end{multline*}

\section{Detailed Results of the Experiments}
 \label{appendix:detailed_results}

 See Table~\ref{table:filter_torque}--\ref{table:filtered_monaa_accel}.

 \begin{table}[tb]
 \centering
 \caption{Filter (\textsc{Torque})}
 \label{table:filter_torque}
  \scalebox{0.75}{
   \begin{filecontents*}{./table/filt-all-torque.tsv}
242808 1 0.121 1796.8 0.2375 3725 242808 0.2165
487021 1 0.246 1791.2 0.402 3709.2 487021 0.3605
732281 1 0.363 1796.4 0.57 3695 732281 0.511
1221149 1 0.6005 1787.4 0.898 3715.2 1221149 0.8065
1830434 1 0.9265 1791.2 1.3095 3724.8 1830434 1.16
2436963 1 1.2075 1801.8 1.718 3696.2 2436963 1.518
3045927 1 1.5405 1796.2 2.143 3692.6 3045927 1.869
3654904 1 1.8725 1787 2.5445 3710.4 3654904 2.2335
4265044 1 2.1575 1790.2 2.943 3721.2 4265044 2.5915
4873207 1 2.442 1788.8 3.4025 3721 4873207 2.978
242808 2 0.0525 1795.2 0.1715 3686 74333 0.2165
487021 2 0.11 1798.2 0.272 3698.6 149805 0.3605
732281 2 0.1595 1793.6 0.369 3693 226107 0.511
1221149 2 0.264 1796 0.555 3711.4 372876 0.8065
1830434 2 0.402 1795.6 0.7945 3709.6 557182 1.16
2436963 2 0.5275 1793.4 1.0355 3712 740611 1.518
3045927 2 0.6665 1802.6 1.266 3706.6 926892 1.869
3654904 2 0.798 1797 1.4955 3722.2 1110315 2.2335
4265044 2 0.933 1791.4 1.75 3703.8 1294699 2.5915
4873207 2 1.059 1797.4 2.042 3716.8 1479893 2.978
242808 3 0.06 1809.6 0.171 3715.8 74274 0.2165
487021 3 0.1195 1806.8 0.272 3711.2 149698 0.3605
732281 3 0.176 1797.8 0.3685 3695.4 225932 0.511
1221149 3 0.291 1803.4 0.561 3696 372592 0.8065
1830434 3 0.4425 1814.4 0.8055 3700.8 556749 1.16
2436963 3 0.5815 1810.8 1.04 3706.2 740047 1.518
3045927 3 0.7375 1802.4 1.2775 3701 926197 1.869
3654904 3 0.8835 1799.8 1.5215 3716.8 1109452 2.2335
4265044 3 1.0365 1809.2 1.7695 3698.4 1293686 2.5915
4873207 3 1.161 1803.6 2.033 3719 1478727 2.978
242808 4 0.057 1808 0.1705 3698.8 74305 0.2165
487021 4 0.11 1801.4 0.2665 3704.4 149743 0.3605
732281 4 0.1635 1796.6 0.367 3692.2 225993 0.511
1221149 4 0.2715 1801.6 0.5555 3697.4 372702 0.8065
1830434 4 0.409 1801 0.7965 3704.2 556904 1.16
2436963 4 0.536 1806 1.027 3705.8 740241 1.518
3045927 4 0.681 1804.2 1.2715 3701.2 926438 1.869
3654904 4 0.816 1807.6 1.5055 3705.6 1109733 2.2335
4265044 4 0.9525 1813.4 1.7445 3722.8 1294033 2.5915
4873207 4 1.0765 1811 1.9925 3707.8 1479095 2.978
242808 5 0.0575 1806 0.171 3688.4 74133 0.2165
487021 5 0.11 1805.4 0.267 3691 149439 0.3605
732281 5 0.1625 1809.8 0.3675 3691.6 225548 0.511
1221149 5 0.273 1812.2 0.5545 3701.8 371902 0.8065
1830434 5 0.4095 1813.8 0.795 3722.8 555670 1.16
2436963 5 0.5375 1809.6 1.03 3707.2 738656 1.518
3045927 5 0.679 1809.8 1.2785 3713 924412 1.869
3654904 5 0.8135 1815.4 1.5135 3697.8 1107300 2.2335
4265044 5 0.956 1821.8 1.7505 3708 1291174 2.5915
4873207 5 1.0795 1820.6 1.985 3720 1475874 2.978
242808 10 0.057 1812.6 0.1705 3703.8 74119 0.2165
487021 10 0.1105 1819.6 0.2695 3701.4 149425 0.3605
732281 10 0.164 1823.8 0.371 3697.8 225513 0.511
1221149 10 0.269 1815.6 0.552 3708.4 371832 0.8065
1830434 10 0.408 1817.2 0.792 3702.8 555572 1.16
2436963 10 0.5345 1816.6 1.0245 3710.6 738537 1.518
3045927 10 0.6805 1824.2 1.278 3703.2 924272 1.869
3654904 10 0.8195 1812.2 1.499 3710 1107153 2.2335
4265044 10 0.956 1816.6 1.7505 3731.6 1290985 2.5915
4873207 10 1.079 1818.8 1.9815 3714.8 1475657 2.978
\end{filecontents*}
\pgfplotstabletypeset[
   sci,
   sci zerofill,
   multicolumn names, 
   columns={[index]0,[index]1,[index]2,[index]3,[index]6},
   display columns/0/.style={
   column name=\begin{tabular}{c}$|w|$\end{tabular}, 
   fixed,fixed zerofill,precision=0,
   },
   display columns/1/.style={
   column name=\begin{tabular}{c}$N$\end{tabular},
   fixed,fixed zerofill,precision=0},
   display columns/2/.style={
   column name=\begin{tabular}{c}Execution Time [s]\end{tabular},
   precision=2},
   display columns/3/.style={
   fixed,precision=2,
   column name=\begin{tabular}{c}Memory Usage [kbyte]\end{tabular}},
   display columns/4/.style={
   fixed,fixed zerofill,precision=0,
   column name=\begin{tabular}{c}$|w'|$\end{tabular}},
   fixed,fixed zerofill,precision=1,
   every head row/.style={
   before row={\toprule}, 
   after row={\midrule} 
   },
   string replace={0}{},
   empty cells with={$< 0.01$},
   every last row/.style={after row=\bottomrule}, 
   ]{./table/filt-all-torque.tsv}}
 \end{table}

 \begin{table}[tb]
 \centering
 \caption{Filter (\textsc{Gear})}
 \label{table:filter_gear}
  \scalebox{0.75}{
   \begin{filecontents*}{./table/filt-all-gear.tsv}
306 1 0 1658.8 0.011 2163.6 201 0.0105
127552 1 0.0465 1779.4 0.18 2250 75666 0.2105
255750 1 0.09 1779 0.356 2258.6 152584 0.419
383168 1 0.132 1785 0.5215 2273.2 229136 0.6175
508756 1 0.1735 1779.4 0.6885 2278.2 301952 0.813
632484 1 0.217 1785.8 0.828 2260 369940 0.993
758500 1 0.258 1779.6 1 2266 444020 1.1925
894692 1 0.3095 1784 1.2115 2261.2 534112 1.4255
1011426 1 0.345 1781.6 1.3245 2255.2 591858 1.576
306 2 0 1657.2 0.0115 2172.6 59 0.0105
127552 2 0.04 1781.2 0.18 2273.6 64245 0.2105
255750 2 0.079 1781.8 0.342 2266.2 129929 0.419
383168 2 0.117 1781 0.5115 2265.8 195143 0.6175
508756 2 0.152 1781.6 0.666 2265.4 256385 0.813
632484 2 0.191 1781.4 0.81 2253.2 312581 0.993
758500 2 0.223 1779 0.962 2267.4 375159 1.1925
894692 2 0.2715 1776.4 1.1695 2257 454755 1.4255
1011426 2 0.3045 1778.4 1.2885 2250.2 500001 1.576
306 3 0 1664.8 0.01 2167.2 59 0.0105
127552 3 0.0415 1785.4 0.1765 2289.8 64245 0.2105
255750 3 0.083 1793.6 0.3425 2279 129929 0.419
383168 3 0.1245 1785.4 0.5105 2265.2 195143 0.6175
508756 3 0.1645 1791 0.667 2280.6 256385 0.813
632484 3 0.2025 1786.6 0.8045 2266.8 312581 0.993
758500 3 0.2415 1791 0.9805 2247.2 375159 1.1925
894692 3 0.2895 1790 1.185 2258.2 454755 1.4255
1011426 3 0.3235 1786.8 1.292 2243.8 500001 1.576
306 4 0 1657.4 0.0105 2155.8 59 0.0105
127552 4 0.04 1778.4 0.176 2273.4 64245 0.2105
255750 4 0.0795 1776.8 0.345 2265.2 129929 0.419
383168 4 0.1215 1783.2 0.5145 2261.2 195143 0.6175
508756 4 0.155 1780.6 0.6655 2273.4 256385 0.813
632484 4 0.1915 1777.2 0.8145 2269 312581 0.993
758500 4 0.2315 1782.6 0.9655 2259.8 375159 1.1925
894692 4 0.273 1778.8 1.183 2247.4 454755 1.4255
1011426 4 0.3065 1784.6 1.2945 2265.6 500001 1.576
306 5 0 1665.6 0.0125 2162 59 0.0105
127552 5 0.04 1789.6 0.175 2250 64245 0.2105
255750 5 0.0795 1791.8 0.346 2258 129929 0.419
383168 5 0.119 1798 0.5085 2266.8 195143 0.6175
508756 5 0.158 1790.6 0.6625 2262.2 256385 0.813
632484 5 0.1915 1792.6 0.8065 2266.4 312581 0.993
758500 5 0.2305 1790.2 0.979 2274.8 375159 1.1925
894692 5 0.2765 1784.8 1.177 2258.4 454755 1.4255
1011426 5 0.307 1789.8 1.294 2262.4 500001 1.576
306 10 0 1663.2 0.01 2160.2 59 0.0105
127552 10 0.04 1783.4 0.1745 2286 64245 0.2105
255750 10 0.081 1781.2 0.34 2265.2 129929 0.419
383168 10 0.119 1788.2 0.5125 2271.6 195143 0.6175
508756 10 0.1575 1782.2 0.6685 2261.6 256385 0.813
632484 10 0.193 1783 0.8085 2282.4 312581 0.993
758500 10 0.2285 1786.2 0.962 2275.2 375159 1.1925
894692 10 0.2785 1785.6 1.174 2249.6 454755 1.4255
1011426 10 0.3055 1783.8 1.304 2265.4 500001 1.576
\end{filecontents*}
\pgfplotstabletypeset[
   sci,
   sci zerofill,
   multicolumn names, 
   columns={[index]0,[index]1,[index]2,[index]3,[index]6},
   display columns/0/.style={
   column name=\begin{tabular}{c}$|w|$\end{tabular}, 
   fixed,fixed zerofill,precision=0,
   },
   display columns/1/.style={
   fixed,fixed zerofill,precision=0,
   column name=\begin{tabular}{c}$N$\end{tabular}},
   display columns/2/.style={
   column name=\begin{tabular}{c}Execution Time [s]\end{tabular},
   precision=2},
   display columns/3/.style={
   fixed,precision=2,
   column name=\begin{tabular}{c}Memory Usage [kbyte]\end{tabular}},
   display columns/4/.style={
   fixed,fixed zerofill,precision=0,
   column name=\begin{tabular}{c}$|w'|$\end{tabular}},
   fixed,fixed zerofill,precision=2,
   every head row/.style={
   before row={\toprule}, 
   after row={\midrule} 
   },
   string replace={0}{},
   empty cells with={$< 0.01$},
   every last row/.style={after row=\bottomrule}, 
   ]{./table/filt-all-gear.tsv}}
 \end{table}

 \begin{table}[tb]
 \centering
 \caption{Filter (\textsc{Accel})}
 \label{table:filter_accel}
  \scalebox{0.75}{
   \begin{filecontents*}{./table/filt-all-accel.tsv}
708 1 0.002 1776.6 0.2 4953.8 677 0.2
218247 1 0.0985 1903.4 0.2955 4976.2 190726 0.3705
436611 1 0.191 1921.8 0.396 4968 380611 0.5395
655237 1 0.286 1913.2 0.5075 4964.4 572777 0.6885
870967 1 0.381 1914.8 0.6715 4995 761240 0.8515
1087411 1 0.4805 1921.4 0.7075 4971 950127 1.019
1304404 1 0.57 1924.4 0.813 4995.8 1140025 1.181
1527632 1 0.666 1922 0.903 5004 1333230 1.344
1739525 1 0.756 1927.6 1.0145 4957.4 1520334 1.497
708 2 0.008 1810.4 0.2 4960.4 287 0.2
218247 2 0.061 1917 0.259 4987.4 106107 0.3705
436611 2 0.1255 1921.8 0.325 4973.4 215710 0.5395
655237 2 0.184 1913 0.3925 4987.2 323056 0.6885
870967 2 0.242 1912.6 0.457 4982 423590 0.8515
1087411 2 0.2965 1933.4 0.5225 4964.6 517707 1.019
1304404 2 0.357 1924.4 0.58 4972.8 621700 1.181
1527632 2 0.425 1918.8 0.649 4974.6 752312 1.344
1739525 2 0.469 1917.2 0.711 4987.4 828584 1.497
708 3 0.0095 1817.4 0.199 4959.2 154 0.2
218247 3 0.0315 1936.4 0.22 4999.8 16226 0.3705
436611 3 0.06 1949.4 0.2465 4966.8 33137 0.5395
655237 3 0.088 1938 0.2755 4970.4 48475 0.6885
870967 3 0.1135 1937.6 0.302 4982.2 63868 0.8515
1087411 3 0.1415 1940.4 0.3295 4987.2 81909 1.019
1304404 3 0.1695 1949.8 0.3605 4967.2 98221 1.181
1527632 3 0.196 1942 0.386 4967.4 113203 1.344
1739525 3 0.2215 1936.8 0.415 4982.2 131265 1.497
708 4 0.01 1820.4 0.199 4949.8 143 0.2
218247 4 0.0295 1935.4 0.2145 4976.8 13678 0.3705
436611 4 0.05 1943 0.236 4975.4 27300 0.5395
655237 4 0.0705 1943.8 0.26 4973.2 40133 0.6885
870967 4 0.0945 1935.6 0.2825 4989 54500 0.8515
1087411 4 0.1155 1943.6 0.3045 4975 66614 1.019
1304404 4 0.139 1941.6 0.327 4957 80141 1.181
1527632 4 0.1605 1935.6 0.3495 4990.4 93248 1.344
1739525 4 0.181 1942.2 0.374 4988.8 106825 1.497
708 5 0.01 1842.8 0.2 4955.2 1 0.2
218247 5 0.024 1945 0.208 4985 2844 0.3705
436611 5 0.0405 1962.2 0.2165 4975.8 6190 0.5395
655237 5 0.06 1957 0.226 4978 8531 0.6885
870967 5 0.08 1952.8 0.239 4994 10509 0.8515
1087411 5 0.0985 1976.6 0.2555 4972.6 13096 1.019
1304404 5 0.1175 1954.6 0.275 4996.8 15790 1.181
1527632 5 0.135 1958.4 0.3005 4983.8 19545 1.344
1739525 5 0.1515 1960.2 0.3105 4983.2 21170 1.497
708 10 0.01 1890.6 0.2 4952.6 1 0.2
218247 10 0.0275 1987.8 0.21 4979.4 2340 0.3705
436611 10 0.043 2016.2 0.218 5001 5287 0.5395
655237 10 0.061 2005.2 0.221 4979 7068 0.6885
870967 10 0.08 2000.8 0.232 4989.8 8514 0.8515
1087411 10 0.098 2009.4 0.2415 4980 10058 1.019
1304404 10 0.1145 2004.4 0.2615 4969.4 12129 1.181
1527632 10 0.1335 2001.8 0.2965 4998.4 16199 1.344
1739525 10 0.1495 2004.6 0.297 4988.4 16277 1.497
\end{filecontents*}
\pgfplotstabletypeset[
   sci,
   sci zerofill,
   multicolumn names, 
   columns={[index]0,[index]1,[index]2,[index]3,[index]6},
   display columns/0/.style={
   column name=\begin{tabular}{c}$|w|$\end{tabular}, 
   fixed,fixed zerofill,precision=0,
   },
   display columns/1/.style={
   fixed,fixed zerofill,precision=0,
   column name=\begin{tabular}{c}$N$\end{tabular}},
   display columns/2/.style={
   column name=\begin{tabular}{c}Execution Time [s]\end{tabular},
   precision=2},
   display columns/3/.style={
   fixed,precision=2,
   column name=\begin{tabular}{c}Memory Usage [kbyte]\end{tabular}},
   display columns/4/.style={
   fixed,fixed zerofill,precision=0,
   column name=\begin{tabular}{c}$|w'|$\end{tabular}},
   fixed,fixed zerofill,precision=2,
   every head row/.style={
   before row={\toprule}, 
   after row={\midrule} 
   },
   string replace={0}{},
   empty cells with={$< 0.01$},
   every last row/.style={after row=\bottomrule}, 
   ]{./table/filt-all-accel.tsv}}
 \end{table}

 \begin{table}[tb]
 \centering
 \caption{Filtered \monaa (\textsc{Torque})}
 \label{table:filtered_monaa_torque}
  \scalebox{0.75}{
   \pgfplotstabletypeset[
   sci,
   sci zerofill,
   multicolumn names, 
   columns={[index]0,[index]1,[index]4,[index]5},
   display columns/0/.style={
   column name=\begin{tabular}{c}$|w|$\end{tabular}, 
   fixed,fixed zerofill,precision=0,
   },
   display columns/1/.style={
   column name=\begin{tabular}{c}$N$\end{tabular},
   fixed,fixed zerofill,precision=0},
   display columns/2/.style={
   column name=\begin{tabular}{c}Execution Time [s]\end{tabular},
   precision=2},
   display columns/3/.style={
   fixed,precision=2,
   column name=\begin{tabular}{c}Memory Usage [kbyte]\end{tabular}},
   fixed,fixed zerofill,precision=2,
   every head row/.style={
   before row={\toprule}, 
   after row={\midrule} 
   },
   string replace={0}{},
   empty cells with={$< 0.01$},
   every last row/.style={after row=\bottomrule}, 
   ]{./table/filt-all-torque.tsv}}
 \end{table}

 \begin{table}[tb]
 \centering
 \caption{Filtered \monaa (\textsc{Gear})}
 \label{table:filtered_monaa_gear}
  \scalebox{0.75}{
   \pgfplotstabletypeset[
   sci,
   sci zerofill,
   multicolumn names, 
   columns={[index]0,[index]1,[index]4,[index]5},
   display columns/0/.style={
   column name=\begin{tabular}{c}$|w|$\end{tabular}, 
   fixed,fixed zerofill,precision=0,
   },
   display columns/1/.style={
   column name=\begin{tabular}{c}$N$\end{tabular},
   fixed,fixed zerofill,precision=0},
   display columns/2/.style={
   column name=\begin{tabular}{c}Execution Time [s]\end{tabular},
   precision=2},
   display columns/3/.style={
   fixed,precision=2,
   column name=\begin{tabular}{c}Memory Usage [kbyte]\end{tabular}},
   fixed,fixed zerofill,precision=2,
   every head row/.style={
   before row={\toprule}, 
   after row={\midrule} 
   },
   string replace={0}{},
   empty cells with={$< 0.01$},
   every last row/.style={after row=\bottomrule}, 
   ]{./table/filt-all-gear.tsv}}
 \end{table}

 \begin{table}[tb]
 \centering
 \caption{Filtered \monaa (\textsc{Accel})}
 \label{table:filtered_monaa_accel}
  \scalebox{0.75}{
   \pgfplotstabletypeset[
   sci,
   sci zerofill,
   multicolumn names, 
   columns={[index]0,[index]1,[index]4,[index]5},
   display columns/0/.style={
   column name=\begin{tabular}{c}$|w|$\end{tabular}, 
   fixed,fixed zerofill,precision=0,
   },
   display columns/1/.style={
   fixed,fixed zerofill,precision=0,
   column name=\begin{tabular}{c}$N$\end{tabular}},
   display columns/2/.style={
   column name=\begin{tabular}{c}Execution Time [s]\end{tabular},
   precision=2},
   display columns/3/.style={
   fixed,precision=2,
   column name=\begin{tabular}{c}Memory Usage [kbyte]\end{tabular}},
   fixed,fixed zerofill,precision=2,
   every head row/.style={
   before row={\toprule}, 
   after row={\midrule} 
   },
   string replace={0}{},
   empty cells with={$< 0.01$},
   every last row/.style={after row=\bottomrule}, 
   ]{./table/filt-all-accel.tsv}}
 \end{table}

\end{document}